%% file: main.tex
\documentclass[runningheads]{llncs}

\input{main-includes.tex}

\usepackage{graphicx}

\begin{document}
\title{An SMT Solver for Regular Expressions and Linear Arithmetic over String Length}
\titlerunning{An SMT Solver for Regexes and String Length}
%
%

\author{Murphy~Berzish\inst{1} \and
	  Mitja~Kulczynski\inst{2} \and
	  Federico~Mora\inst{3} \and
	  Florin~Manea\inst{4} \and
	  Joel~D.~Day\inst{5} \and
	  Dirk~Nowotka\inst{2}\and   
	  Vijay~Ganesh\inst{1}
	}
\authorrunning{M. Berzish et al.}
\institute{University of Waterloo, Waterloo, Canada\and
	       Kiel University, Kiel, Germany \and
		   University of California, Berkeley, USA \and
 		   University of G\"ottingen and Campus-Institute Data Science, G\"ottingen, Germany \and
  		   Loughborough University, Loughborough, UK
}

\input{plots.tex}

\maketitle

\input{abstract.tex}

\section{Introduction}
\input{introduction.tex}

\section{Preliminaries}
\input{preliminaries.tex}

\section{Length-Aware Regular Expression Algorithm}
\input{algorithm-2020.tex}

\section{Empirical Results}
\input{experimental-results.tex}

\section{Related Work}
\input{related-work.tex}

\section{Conclusions and Future Work}
\input{conclusion.tex}

\bibliographystyle{splncs04}
\bibliography{main,cav15}

\input{appendix_theory}

\end{document}

%% file: main-includes.tex
\usepackage{amssymb}
\usepackage{mathtools}
\usepackage{color}
\usepackage{wrapfig}
\usepackage{comment}
\usepackage{thmtools} 
\usepackage{thm-restate}

\usepackage[ruled, lined, linesnumbered]{algorithm2e}
\usepackage{multirow}
\usepackage{listings}
\usepackage{multirow}
\usepackage{multicol}
\usepackage{graphicx}
\usepackage{paralist}
\usepackage[title]{appendix}
\usepackage{tikz}
\usepackage{todonotes} 
\usepackage{makecell}
\usetikzlibrary{positioning, shapes,patterns,shadows.blur, arrows,decorations.text,arrows,automata,shadows,patterns,chains,bending,arrows.meta}

\definecolor{tfg}{HTML}{EEEEEE}
\definecolor{tbg}{HTML}{333333}
\tikzset{%
	initial text={},
	unit/.style={draw=tfg,line width=0.15pt,rectangle, rounded corners=3pt,line width=0.15pt, color=tbg,fill=tfg!0,font=\tiny,text width=1.5cm, align=center, anchor=center,rounded corners=1pt},
}
\usepackage{pgfplots}
\pgfplotsset{compat=1.17}
\usepgfplotslibrary{fillbetween}
\usepackage{array}

\SetKwInOut{Input}{Input}
\SetKwInOut{Output}{Output}

\usepackage{xspace}
\newcommand{\toolname}{Z3str3RE\xspace}
\newcommand{\benchmarkgen}{StringFuzz-regex-generated\xspace}
\newcommand{\benchmarkxform}{StringFuzz-regex-transformed\xspace}
\newcommand{\benchmarkcollected}{RegEx-Collected\xspace}
\newcommand{\benchmarkregexlib}{AutomatArk\xspace}

\newcommand{\vsep}{\hspace{3mm} | \hspace{3mm}}
\newcommand{\totalinstances}{57256}

\definecolor{colourCVC4}{HTML}{25333D}
\definecolor{colourZ3Seq}{HTML}{0065AB}
\definecolor{colourOSTRICH}{HTML}{8939AD}
\definecolor{colourZ3-Trau}{HTML}{007E7A}
\definecolor{colourZ3str3}{HTML}{CD3517}
\definecolor{colourZ3str3RE}{HTML}{318700}

%% file: plots.tex
\definecolor{colourCVC4}{HTML}{25333D}
\definecolor{colourZ3Seq}{HTML}{0065AB}
\definecolor{colourOSTRICH}{HTML}{8939AD}
\definecolor{colourZ3-Trau}{HTML}{007E7A}
\definecolor{colourZ3str3}{HTML}{CD3517}
\definecolor{colourZ3str3RE}{HTML}{318700}
\newcommand{\cactusTotalCroped}{
\pgfplotsset{scaled x ticks=false}\pgfplotsset{scaled y ticks=false}

\begin{tikzpicture}\begin{axis}[xmin=-1000,ymax=27500,xlabel=Solved instances,ylabel=Time (seconds),,legend columns=2,legend style={at={(1.05,0.9)},nodes={scale=1, transform shape}, fill=white,anchor=east,align=center },axis line style={draw=none}, xtick pos=left, ytick pos=left, ymajorgrids=true, legend style={draw=none},x post scale=2,y post scale=1.25]

\addplot[name path=pathCVC4 , colourCVC4, line width=1.5pt,dashed] coordinates {(1,0.007)(282,3.0111669626998223)(845,9.968207815275312)(1408,17.824362344582592)(1971,26.368163410301953)(2534,35.50706039076376)(3097,45.21454351687389)(3660,55.404946714031965)(4223,66.06789165186501)(4786,77.204)(5349,88.72730373001777)(5912,100.6566269982238)(6475,113.0046944937833)(7038,125.733)(7601,138.76817406749558)(8164,152.22430373001777)(8727,166.036)(9290,180.20451687388987)(9853,194.792)(10416,209.712)(10979,224.95896447602132)(11542,240.623)(12105,256.6032895204263)(12668,272.9372397868561)(13231,289.647)(13794,306.61988632326825)(14357,324.014)(14920,341.649648312611)(15483,359.65858081705153)(16046,378.016)(16609,396.63141740674956)(17172,415.658)(17735,434.9555435168739)(18298,454.64834103019535)(18861,474.713)(19424,495.03109413854355)(19987,515.768)(20550,536.7744706927176)(21113,558.1620444049734)(21676,579.914)(22239,601.9388969804618)(22802,624.386)(23365,647.1810373001776)(23928,670.3343747779752)(24491,693.91)(25054,717.8193037300177)(25617,742.1055488454706)(26180,766.8179271758437)(26743,791.934)(27306,817.379351687389)(27869,843.2758703374777)(28432,869.6300444049734)(28995,896.449)(29558,923.6814724689166)(30121,951.3977335701599)(30684,979.5864635879218)(31247,1008.2695612788632)(31810,1037.4889662522203)(32373,1067.2572291296624)(32936,1097.5746074600354)(33499,1128.498374777975)(34062,1160.0405222024867)(34625,1192.2046145648314)(35188,1225.0690870337478)(35751,1258.6243108348133)(36314,1292.8908827708703)(36877,1327.9033765541742)(37440,1363.7590710479574)(38003,1400.491401420959)(38566,1438.168781527531)(39129,1476.814408525755)(39692,1516.5366536412077)(40255,1557.39691651865)(40818,1599.4343623445825)(41381,1642.8413179396093)(41944,1687.7198845470691)(42507,1734.1310763765541)(43070,1782.2261492007106)(43633,1832.183651865009)(44196,1884.266728241563)(44759,1938.7338028419183)(45322,1996.0026429840143)(45885,2056.664651865009)(46448,2121.022310834813)(47011,2189.6774156305505)(47574,2263.6465097690943)(48137,2344.3423783303733)(48700,2432.4105062166964)(49263,2529.2283854351685)(49826,2636.6197300177623)(50389,2758.4383108348134)(50952,2901.2783019538188)(51515,3074.0632095914743)(52078,3305.67513321492)(52641,3632.2230568383657)(53204,4033.809079928952)(53767,4706.927253996448)(54330,7083.351701598579)(54893,12112.36734813499)(55190,16337.284129032258)(55207,16645.468)};
\addlegendentry{CVC4}

\addplot[name path=pathZ3Seq , colourZ3Seq, line width=1.5pt,densely dotted] coordinates {(1,0.017)(270.5,5.797468518518518)(810.5,18.306733333333334)(1350.5,31.637001851851853)(1890.5,45.598251851851856)(2430.5,60.10216851851852)(2970.5,75.09856111111111)(3510.5,90.5605)(4050.5,106.45991851851852)(4590.5,122.76698333333334)(5130.5,139.4890388888889)(5670.5,156.622)(6210.5,174.0932685185185)(6750.5,191.9479185185185)(7290.5,210.208)(7830.5,228.7747185185185)(8370.5,247.6961685185185)(8910.5,267.007)(9450.5,286.57825185185186)(9990.5,306.53447777777774)(10530.5,326.816)(11070.5,347.3891111111111)(11610.5,368.3655)(12150.5,389.60516666666666)(12690.5,411.20060185185184)(13230.5,433.1365)(13770.5,455.3265518518519)(14310.5,477.919)(14850.5,500.77622777777776)(15390.5,523.9885018518519)(15930.5,547.523)(16470.5,571.3606018518518)(17010.5,595.6025)(17550.5,620.1243611111112)(18090.5,644.9761222222222)(18630.5,670.1955)(19170.5,695.7166944444444)(19710.5,721.6269)(20250.5,747.9145)(20790.5,774.5397833333333)(21330.5,801.5565018518519)(21870.5,828.9755)(22410.5,856.7213444444444)(22950.5,884.8746018518518)(23490.5,913.4466111111111)(24030.5,942.431)(24570.5,971.8220185185186)(25110.5,1001.6552611111111)(25650.5,1031.941261111111)(26190.5,1062.6743018518519)(26730.5,1093.8498333333332)(27270.5,1125.4825)(27810.5,1157.5772333333334)(28350.5,1190.2060944444445)(28890.5,1223.342561111111)(29430.5,1257.013351851852)(29970.5,1291.2622351851853)(30510.5,1326.1016666666667)(31050.5,1361.5895518518519)(31590.5,1397.7107555555554)(32130.5,1434.4738111111112)(32670.5,1471.9572685185185)(33210.5,1510.2590277777776)(33750.5,1549.357562962963)(34290.5,1589.2914074074074)(34830.5,1630.1085666666665)(35370.5,1671.832062962963)(35910.5,1714.6123537037038)(36450.5,1758.4755)(36990.5,1803.4987796296296)(37530.5,1849.8234)(38070.5,1897.4924055555557)(38610.5,1946.689024074074)(39150.5,1997.5361907407407)(39690.5,2050.2358740740738)(40230.5,2104.8948314814816)(40770.5,2161.8325685185187)(41310.5,2221.3171018518515)(41850.5,2283.6397796296296)(42390.5,2349.1968814814813)(42930.5,2418.4826333333335)(43470.5,2492.3227037037036)(44010.5,2571.668164814815)(44550.5,2657.109848148148)(45090.5,2749.1602055555554)(45630.5,2848.481001851852)(46170.5,2955.820507407407)(46710.5,3072.791301851852)(47250.5,3202.6094629629633)(47790.5,3350.0469777777776)(48330.5,3523.036977777778)(48870.5,3731.5543814814814)(49410.5,3991.0068240740743)(49950.5,4321.625137037037)(50490.5,4767.427335185185)(51030.5,5439.391464814815)(51570.5,6559.0604370370365)(52110.5,8606.3676)(52650.5,13046.453448148148)(52940,17199.009897435895)(52961,17587.844)};
\addlegendentry{Z3Seq}

\addplot[name path=pathOSTRICH , colourOSTRICH, line width=1.5pt,dotted] coordinates {(1,0.741)(213.5,253.77244835680753)(639.5,838.5066314553991)(1065.5,1485.46291314554)(1491.5,2179.2085187793427)(1917.5,2906.372516431925)(2343.5,3659.9279483568075)(2769.5,4435.3492957746475)(3195.5,5231.1529037558685)(3621.5,6045.613194835681)(4047.5,6876.266072769953)(4473.5,7721.9782699530515)(4899.5,8581.362068075117)(5325.5,9452.97274882629)(5751.5,10335.134373239438)(6177.5,11228.048934272301)(6603.5,12131.827274647887)(7029.5,13046.23066431925)(7455.5,13970.923758215962)(7881.5,14905.723244131455)(8307.5,15850.64652112676)(8733.5,16805.796192488266)(9159.5,17770.844711267604)(9585.5,18746.34308920188)(10011.5,19731.605875586854)(10437.5,20727.260396713613)(10863.5,21734.545056338025)(11289.5,22754.631730046945)(11715.5,23788.160896713616)(12141.5,24835.54532159624)(12567.5,25897.952122065726)(12993.5,26977.07145305164)(13419.5,28074.41346478873)(13845.5,29191.308708920187)(14271.5,30328.970295774645)(14697.5,31489.02130516432)(15123.5,32673.066943661972)(15549.5,33881.79666666666)(15975.5,35113.15557042254)(16401.5,36366.68859389672)(16827.5,37643.07607276995)(17253.5,38941.16469718309)(17679.5,40260.028737089204)(18105.5,41599.608678403754)(18531.5,42959.86288967136)(18957.5,44340.549328638495)(19383.5,45741.648978873236)(19809.5,47164.03138497653)(20235.5,48608.79976056338)(20661.5,50076.0596971831)(21087.5,51566.51402347418)(21513.5,53079.74507746479)(21939.5,54615.06131924882)(22365.5,56171.84766666666)(22791.5,57750.029427230045)(23217.5,59349.435173708924)(23643.5,60970.158349765254)(24069.5,62612.05390610328)(24495.5,64274.86588497653)(24921.5,65957.65288732394)(25347.5,67661.02898122066)(25773.5,69385.78774882629)(26199.5,71131.79469483568)(26625.5,72899.49501173709)(27051.5,74688.36697887324)(27477.5,76498.89917840375)(27903.5,78332.68050469483)(28329.5,80190.55502816902)(28755.5,82074.66961267605)(29181.5,83985.77750704226)(29607.5,85925.48217370892)(30033.5,87896.74588497652)(30459.5,89900.21355868546)(30885.5,91936.8474741784)(31311.5,94011.70163615025)(31737.5,96131.0384600939)(32163.5,98298.54632629108)(32589.5,100513.73207276996)(33015.5,102777.78368779343)(33441.5,105096.86781924883)(33867.5,107479.31934507043)(34293.5,109939.0126056338)(34719.5,112486.02324647887)(35145.5,115129.45070657277)(35571.5,117875.99636854461)(35997.5,120737.13184976525)(36423.5,123734.21921361503)(36849.5,126876.42945774648)(37275.5,130168.91105399061)(37701.5,133624.06801643193)(38127.5,137274.79945305164)(38553.5,141150.20496244132)(38979.5,145286.22809859153)(39405.5,149724.2365657277)(39831.5,154494.49541784037)(40257.5,159615.38779342722)(40683.5,165214.34435915493)(41109.5,171461.91916197183)(41535.5,178608.96715258216)(41749.5,182621.217)(41752,182651.276)};
\addlegendentry{OSTRICH}

\addplot[name path=pathZ3-Trau , colourZ3-Trau, line width=1.5pt,dash pattern={on 7pt off 2pt on 1pt off 3pt}] coordinates {(1,0.019)(204.5,4.8718063725490195)(612.5,15.249235294117646)(1020.5,26.20266176470588)(1428.5,37.56254411764706)(1836.5,49.2655)(2244.5,61.25676470588235)(2652.5,73.53131617647058)(3060.5,86.1125)(3468.5,98.90091911764706)(3876.5,111.96217647058823)(4284.5,125.3065)(4692.5,138.87109068627453)(5100.5,152.7458700980392)(5508.5,166.9315)(5916.5,181.38154411764705)(6324.5,196.10075980392156)(6732.5,211.1235)(7140.5,226.41022303921568)(7548.5,241.96727450980393)(7956.5,257.8305)(8364.5,273.9116323529412)(8772.5,290.2673161764706)(9180.5,306.9285)(9588.5,323.8311397058823)(9996.5,340.9971764705883)(10404.5,358.4655)(10812.5,376.18321568627454)(11220.5,394.1617573529412)(11628.5,412.4405)(12036.5,430.9493382352941)(12444.5,449.73387009803923)(12852.5,468.8155)(13260.5,488.17173529411764)(13668.5,507.82124509803924)(14076.5,527.7834411764705)(14484.5,548.052)(14892.5,568.5842450980391)(15300.5,589.4239852941176)(15708.5,610.5762230392156)(16116.5,632.0095)(16524.5,653.7706200980391)(16932.5,675.8692598039215)(17340.5,698.3292647058825)(17748.5,721.1485514705882)(18156.5,744.311818627451)(18564.5,767.813)(18972.5,791.6610147058824)(19380.5,815.8813382352942)(19788.5,840.4789852941176)(20196.5,865.48475)(20604.5,890.8939852941176)(21012.5,916.7467647058824)(21420.5,943.0284411764706)(21828.5,969.7811691176471)(22236.5,997.0085906862745)(22644.5,1024.6960073529413)(23052.5,1052.8713897058824)(23460.5,1081.5847352941175)(23868.5,1110.8828186274511)(24276.5,1140.7694411764705)(24684.5,1171.3359754901962)(25092.5,1202.679044117647)(25500.5,1234.8166470588237)(25908.5,1267.8463676470587)(26316.5,1301.897480392157)(26724.5,1337.1078799019608)(27132.5,1373.519593137255)(27540.5,1411.2236225490194)(27948.5,1450.2845612745098)(28356.5,1491.007112745098)(28764.5,1533.68693872549)(29172.5,1578.5677916666668)(29580.5,1625.7606274509806)(29988.5,1675.5032769607844)(30396.5,1728.1060612745098)(30804.5,1783.6573578431373)(31212.5,1842.8395073529412)(31620.5,1906.2164828431373)(32028.5,1974.1974166666666)(32436.5,2047.9800857843138)(32844.5,2128.1046985294115)(33252.5,2216.4344313725487)(33660.5,2315.5847156862746)(34068.5,2429.6318137254902)(34476.5,2567.158524509804)(34884.5,2739.079973039216)(35292.5,2946.474156862745)(35700.5,3203.3692892156864)(36108.5,3544.664102941176)(36516.5,4015.6255661764703)(36924.5,4648.869460784314)(37332.5,5432.128409313726)(37740.5,6412.605977941177)(38148.5,7714.302416666667)(38556.5,9481.599551470588)(38964.5,11730.883286764705)(39372.5,14819.870875)(39780.5,19857.07048529412)(39989,23489.695)(39995,23589.419)};
\addlegendentry{Z3-Trau}

\addplot[name path=pathZ3str3 , colourZ3str3, line width=1.5pt,dash dot] coordinates {(1,0.012)(236,4.206135881104034)(707,13.492554140127387)(1178,23.351445859872612)(1649,33.622)(2120,44.13931422505308)(2591,54.971050955414015)(3062,66.114)(3533,77.44172611464967)(4004,89.107)(4475,100.94238428874735)(4946,113.131)(5417,125.53805095541402)(5888,138.2488619957537)(6359,151.229)(6830,164.46448619957536)(7301,178.052)(7772,191.94103184713376)(8243,206.1551210191083)(8714,220.71974097664543)(9185,235.595)(9656,250.75112101910827)(10127,266.25938853503186)(10598,282.091)(11069,298.2027834394904)(11540,314.66283651804673)(12011,331.446)(12482,348.5103651804671)(12953,365.94068789808915)(13424,383.737)(13895,401.79031422505307)(14366,420.20168152866245)(14837,438.97512101910826)(15308,458.082)(15779,477.4856963906582)(16250,497.24057961783444)(16721,517.325)(17192,537.7069299363058)(17663,558.4416687898089)(18134,579.5345732484077)(18605,601.002)(19076,622.7947197452229)(19547,644.9476963906582)(20018,667.4652951167728)(20489,690.353)(20960,713.6241974522293)(21431,737.2609702760085)(21902,761.3034012738854)(22373,785.7655881104034)(22844,810.6501910828026)(23315,835.9602802547771)(23786,861.6912951167728)(24257,887.8873715498938)(24728,914.6012993630574)(25199,941.79325477707)(25670,969.51183014862)(26141,997.7679363057325)(26612,1026.602280254777)(27083,1056.0078428874735)(27554,1086.0264861995754)(28025,1116.7280445859872)(28496,1148.1238280254777)(28967,1180.196050955414)(29438,1213.0275031847134)(29909,1246.629118895966)(30380,1281.1001210191082)(30851,1316.4649256900211)(31322,1352.793847133758)(31793,1390.1873694267515)(32264,1428.7451740976646)(32735,1468.5681295116772)(33206,1509.7862250530784)(33677,1552.6686518046708)(34148,1597.2995053078555)(34619,1643.792101910828)(35090,1692.3383163481953)(35561,1743.1576199575372)(36032,1796.4282250530784)(36503,1852.481789808917)(36974,1911.749008492569)(37445,1974.561354564756)(37916,2041.160280254777)(38387,2111.402295116773)(38858,2185.3015031847135)(39329,2264.1268938428875)(39800,2349.5750955414014)(40271,2443.663585987261)(40742,2548.85898089172)(41213,2668.5742186836515)(41684,2805.1200997876854)(42155,2965.154978768577)(42626,3155.4265350318474)(43097,3389.988016985138)(43568,3710.9221634819532)(44039,4258.287101910828)(44510,5110.215972399151)(44981,6309.254949044586)(45452,7903.1780212314225)(45901,10731.045782201405)(46116,13423.939)};
\addlegendentry{Z3str3}

\addplot[name path=pathZ3str3RE , colourZ3str3RE, line width=1.5pt] coordinates {(1,0.007)(287,2.701500872600349)(860,9.067654450261779)(1433,16.179886561954625)(2006,23.65148167539267)(2579,31.53)(3152,39.552)(3725,47.60667190226876)(4298,56.076)(4871,64.671)(5444,73.266)(6017,81.86115881326353)(6590,90.756)(7163,99.924)(7736,109.092)(8309,118.26)(8882,127.428)(9455,136.74232635253054)(10028,146.46)(10601,156.201)(11174,165.942)(11747,175.683)(12320,185.424)(12893,195.28349214659687)(13466,205.561)(14039,215.875)(14612,226.189)(15185,236.503)(15758,246.817)(16331,257.2135095986039)(16904,268.039)(17477,278.926)(18050,289.813)(18623,300.7)(19196,311.60278534031414)(19769,322.895)(20342,334.355)(20915,345.815)(21488,357.275)(22061,369.00913612565444)(22634,381.042)(23207,393.075)(23780,405.1190436300174)(24353,417.54)(24926,430.146)(25499,442.762462478185)(26072,455.754)(26645,468.933)(27218,482.2324293193717)(27791,495.949)(28364,509.71631937172776)(28937,523.872)(29510,538.1972984293194)(30083,552.827)(30656,567.7278848167539)(31229,582.967)(31802,598.4752094240838)(32375,614.402)(32948,630.5973717277486)(33521,647.2138289703316)(34094,664.251)(34667,681.6838062827225)(35240,699.5706719022688)(35813,717.9466282722514)(36386,736.8281989528796)(36959,756.235)(37532,776.1857591623036)(38105,796.7000366492147)(38678,817.8235602094242)(39251,839.581410122164)(39824,862.0013507853403)(40397,885.1013490401396)(40970,908.9923717277486)(41543,933.6900959860385)(42116,959.2246596858639)(42689,985.625219895288)(43262,1013.0012547993019)(43835,1041.5429109947645)(44408,1071.396520069808)(44981,1102.7225427574172)(45554,1135.8555776614312)(46127,1171.2148359511343)(46700,1208.839642233857)(47273,1248.7797242582897)(47846,1291.3442286212914)(48419,1336.9697207678882)(48992,1386.4005287958116)(49565,1441.001132635253)(50138,1502.2777486910995)(50711,1573.2778883071553)(51284,1651.8690261780105)(51857,1736.4811204188481)(52430,1835.6542897033157)(53003,1972.8160261780106)(53576,2142.884005235602)(54149,2366.9972739965096)(54722,2705.7645462478185)(55295,3319.7226579406633)(55868,4681.389465968587)(56156,6222.517)(56159,6260.822)};
\addlegendentry{Z3str3RE}

\path[name path=axisCVC4] (axis cs:0,0) -- (axis cs:55207,0);
\addplot [thick,color=colourCVC4,fill=colourCVC4,fill opacity=0.1] fill between [of=pathCVC4 and axisCVC4];
\path[name path=axisZ3Seq] (axis cs:0,0) -- (axis cs:52961,0);
\addplot [thick,color=colourZ3Seq,fill=colourZ3Seq,fill opacity=0.1] fill between [of=pathZ3Seq and axisZ3Seq];
\path[name path=axisOSTRICH] (axis cs:0,0) -- (axis cs:41752,0);
\addplot [thick,color=colourOSTRICH,fill=colourOSTRICH,fill opacity=0.1] fill between [of=pathOSTRICH and axisOSTRICH];
\path[name path=axisZ3-Trau] (axis cs:0,0) -- (axis cs:39995,0);
\addplot [thick,color=colourZ3-Trau,fill=colourZ3-Trau,fill opacity=0.1] fill between [of=pathZ3-Trau and axisZ3-Trau];
\path[name path=axisZ3str3] (axis cs:0,0) -- (axis cs:46116,0);
\addplot [thick,color=colourZ3str3,fill=colourZ3str3,fill opacity=0.1] fill between [of=pathZ3str3 and axisZ3str3];
\path[name path=axisZ3str3RE] (axis cs:0,0) -- (axis cs:56159,0);
\addplot [thick,color=colourZ3str3RE,fill=colourZ3str3RE,fill opacity=0.1] fill between [of=pathZ3str3RE and axisZ3str3RE];
\end{axis}\end{tikzpicture}}


\newcommand{\cactusHeuristics}{
\pgfplotsset{scaled x ticks=false}\pgfplotsset{scaled y ticks=false}
\definecolor{colourZ3str3RE-base}{HTML}{318700}
\definecolor{colourZ3str3RE-ali}{HTML}{333333}
\definecolor{colourZ3str3RE-li}{HTML}{666666}
\definecolor{colourZ3str3RE-asi}{HTML}{999999}
\definecolor{colourZ3str3RE-psh}{HTML}{25333D}
\definecolor{colourZ3str3RE-none}{HTML}{939393}
\begin{tikzpicture}\begin{axis}[xmin=46000,xmax=56500,xlabel=Solved instances,ylabel=Time (seconds),,legend columns=1,legend style={nodes={scale=1.25, transform shape}, fill=white,anchor=east,align=center },axis line style={draw=none}, xtick pos=left, ytick pos=left, ymajorgrids=true, legend style={at={(0.05,0.645)},anchor=west,draw=none},x post scale=2,y post scale=1.25,ymin=0]

\addplot[name path=pathZ3str3RE-none , colourZ3str3RE-none, line width=1.5pt, dashed] coordinates {(45501,1937.545)(45539.5,1949.208064102564)(45617.5,1972.485282051282)(45695.5,1996.1090128205128)(45773.5,2020.0363974358975)(45851.5,2044.227141025641)(45929.5,2068.6844615384616)(46007.5,2093.4213974358977)(46085.5,2118.465935897436)(46163.5,2143.8351153846156)(46241.5,2169.508371794872)(46319.5,2195.4937564102565)(46397.5,2221.929294871795)(46475.5,2248.9113205128206)(46553.5,2276.5996923076923)(46631.5,2305.0878333333335)(46709.5,2334.3298846153843)(46787.5,2364.320205128205)(46865.5,2395.1592435897437)(46943.5,2426.9562564102566)(47021.5,2459.9096025641024)(47099.5,2493.9922564102567)(47177.5,2529.302423076923)(47255.5,2565.8428333333336)(47333.5,2603.431217948718)(47411.5,2642.1928846153846)(47489.5,2682.3608076923074)(47567.5,2724.072487179487)(47645.5,2767.192794871795)(47723.5,2811.890512820513)(47801.5,2858.0847564102564)(47879.5,2905.6051153846156)(47957.5,2954.638769230769)(48035.5,3005.6814487179486)(48113.5,3058.5536794871796)(48191.5,3113.153435897436)(48269.5,3169.3019102564103)(48347.5,3226.8449615384616)(48425.5,3285.91)(48503.5,3347.460346153846)(48581.5,3411.9350512820515)(48659.5,3479.4209487179487)(48737.5,3550.065769230769)(48815.5,3624.468064102564)(48893.5,3702.2823974358976)(48971.5,3783.797987179487)(49049.5,3869.508769230769)(49127.5,3959.556128205128)(49205.5,4054.238141025641)(49283.5,4153.890153846154)(49361.5,4259.039858974359)(49439.5,4369.861871794872)(49517.5,4485.768987179487)(49595.5,4608.4571153846155)(49673.5,4739.035448717949)(49751.5,4877.609102564103)(49829.5,5023.95191025641)(49907.5,5177.977923076923)(49985.5,5335.687397435897)(50063.5,5495.308076923077)(50141.5,5657.802051282051)(50219.5,5824.3754487179485)(50297.5,5996.777576923077)(50375.5,6175.135205128205)(50453.5,6359.824782051282)(50531.5,6550.910525641026)(50609.5,6748.741333333333)(50687.5,6952.975474358974)(50765.5,7164.23408974359)(50843.5,7385.031871794872)(50921.5,7615.670487179487)(50999.5,7856.970692307692)(51077.5,8108.3075512820515)(51155.5,8371.709525641027)(51233.5,8648.476256410257)(51311.5,8939.350294871794)(51389.5,9246.032717948718)(51467.5,9564.455653846155)(51545.5,9894.301653846154)(51623.5,10241.5995)(51701.5,10606.609884615384)(51779.5,10991.753307692308)(51857.5,11397.311935897436)(51935.5,11832.11617948718)(52013.5,12301.28076923077)(52091.5,12796.505923076924)(52169.5,13320.328974358974)(52247.5,13879.541141025642)(52325.5,14475.059948717948)(52403.5,15120.315076923076)(52481.5,15819.443615384616)(52559.5,16579.816102564102)(52637.5,17416.512935897437)(52715.5,18342.970423076924)(52793.5,19357.631846153847)(52871.5,20456.041051282053)(52949.5,21638.749076923075)(53027.5,22919.449423076923)(53136,24906.87)};
\addlegendentry{All heuristics off}

\addplot[name path=pathZ3str3RE-li , colourZ3str3RE-li, line width=1.5pt,densely dotted] coordinates {(45501,2219.528)(45541.5,2232.5554634146342)(45623.5,2258.519975609756)(45705.5,2284.7976097560977)(45787.5,2311.3772804878045)(45869.5,2338.2625731707317)(45951.5,2365.4080609756097)(46033.5,2392.8391951219514)(46115.5,2420.6192682926826)(46197.5,2448.8607560975606)(46279.5,2477.6683780487806)(46361.5,2507.01556097561)(46443.5,2537.0052439024394)(46525.5,2567.8027195121954)(46607.5,2599.317073170732)(46689.5,2631.7398048780487)(46771.5,2665.129573170732)(46853.5,2699.7018414634144)(46935.5,2735.8377317073173)(47017.5,2773.5748658536586)(47099.5,2812.5481585365856)(47181.5,2852.5632682926826)(47263.5,2893.6607804878045)(47345.5,2936.2521341463416)(47427.5,2980.7869024390243)(47509.5,3027.592524390244)(47591.5,3076.384231707317)(47673.5,3126.6797439024394)(47755.5,3178.637609756098)(47837.5,3232.8062317073172)(47919.5,3289.3847926829267)(48001.5,3347.8776707317074)(48083.5,3407.894743902439)(48165.5,3469.5763292682927)(48247.5,3533.2128780487806)(48329.5,3599.8341951219513)(48411.5,3669.9052439024395)(48493.5,3743.2349268292683)(48575.5,3820.584731707317)(48657.5,3902.7476951219514)(48739.5,3990.225426829268)(48821.5,4083.8869146341463)(48903.5,4183.611670731708)(48985.5,4289.4252682926835)(49067.5,4400.812390243903)(49149.5,4517.129182926829)(49231.5,4638.1135)(49313.5,4764.590365853659)(49395.5,4897.374487804878)(49477.5,5038.346182926829)(49559.5,5186.367878048781)(49641.5,5340.3942439024395)(49723.5,5502.780036585365)(49805.5,5674.392451219512)(49887.5,5853.6014512195125)(49969.5,6039.9874634146345)(50051.5,6232.871804878048)(50133.5,6433.904)(50215.5,6643.428536585366)(50297.5,6861.989804878049)(50379.5,7091.677719512195)(50461.5,7333.387439024391)(50543.5,7590.041792682927)(50625.5,7862.698060975609)(50707.5,8150.572658536585)(50789.5,8456.673134146342)(50871.5,8783.612841463415)(50953.5,9133.562707317073)(51035.5,9504.984621951218)(51117.5,9894.291853658537)(51199.5,10300.21057317073)(51281.5,10718.476756097562)(51363.5,11148.200280487805)(51445.5,11594.249012195121)(51527.5,12063.468548780487)(51609.5,12555.257231707317)(51691.5,13068.272280487805)(51773.5,13603.349231707318)(51855.5,14156.451146341464)(51937.5,14728.371060975609)(52019.5,15319.553548780488)(52101.5,15928.790646341464)(52183.5,16556.995829268293)(52265.5,17206.524231707317)(52347.5,17877.21398780488)(52429.5,18571.866707317073)(52511.5,19295.35032926829)(52593.5,20049.697231707316)(52675.5,20839.564012195122)(52757.5,21669.183451219513)(52839.5,22543.089914634147)(52921.5,23462.49148780488)(53003.5,24443.021926829268)(53085.5,25496.653573170734)(53167.5,26628.51001219512)(53249.5,27834.880268292683)(53331.5,29109.719560975613)(53413.5,30454.79632926829)(53495.5,31898.13237804878)(53571,33316.008)};
\addlegendentry{Lazy intersections off}

\addplot[name path=pathZ3str3RE-psh , colourZ3str3RE-psh, line width=1.5pt,dotted] coordinates {(45501,1419.383)(45552.5,1423.7934326923078)(45656.5,1432.591375)(45760.5,1441.5035)(45864.5,1450.5463076923077)(45968.5,1459.7152788461538)(46072.5,1469.034951923077)(46176.5,1478.5107307692308)(46280.5,1488.138625)(46384.5,1497.9269423076923)(46488.5,1507.8720192307692)(46592.5,1517.9382788461537)(46696.5,1528.1302788461537)(46800.5,1538.443375)(46904.5,1548.8941442307691)(47008.5,1559.4995384615386)(47112.5,1570.3077788461537)(47216.5,1581.330173076923)(47320.5,1592.573826923077)(47424.5,1604.050048076923)(47528.5,1615.7809807692308)(47632.5,1627.7879134615384)(47736.5,1640.0900384615386)(47840.5,1652.7046153846154)(47944.5,1665.6359038461537)(48048.5,1678.854875)(48152.5,1692.4196730769231)(48256.5,1706.3606057692307)(48360.5,1720.5795192307692)(48464.5,1734.9978076923078)(48568.5,1749.5676057692308)(48672.5,1764.2874038461537)(48776.5,1779.1562788461538)(48880.5,1794.1675673076923)(48984.5,1809.3586153846154)(49088.5,1824.7781634615385)(49192.5,1840.4569038461539)(49296.5,1856.3977403846154)(49400.5,1872.5877596153846)(49504.5,1889.0524807692307)(49608.5,1905.8480961538462)(49712.5,1922.995125)(49816.5,1940.4560961538461)(49920.5,1958.162548076923)(50024.5,1976.1585)(50128.5,1994.5546634615384)(50232.5,2013.6268461538461)(50336.5,2033.6339423076922)(50440.5,2055.340173076923)(50544.5,2079.2371346153845)(50648.5,2105.568230769231)(50752.5,2134.3035)(50856.5,2164.492221153846)(50960.5,2195.4740673076926)(51064.5,2227.1841826923073)(51168.5,2259.5241153846155)(51272.5,2292.3851923076927)(51376.5,2325.6923846153845)(51480.5,2359.5020096153844)(51584.5,2393.8033365384613)(51688.5,2428.721519230769)(51792.5,2464.95475)(51896.5,2503.030923076923)(52000.5,2543.3348365384613)(52104.5,2586.420778846154)(52208.5,2633.5606346153845)(52312.5,2683.5641634615386)(52416.5,2736.6318846153845)(52520.5,2794.800298076923)(52624.5,2857.1524615384615)(52728.5,2923.0481826923074)(52832.5,2994.140855769231)(52936.5,3069.016673076923)(53040.5,3146.8523365384617)(53144.5,3230.531605769231)(53248.5,3323.202221153846)(53352.5,3425.9156346153845)(53456.5,3539.631605769231)(53560.5,3665.352701923077)(53664.5,3801.5771923076927)(53768.5,3947.821076923077)(53872.5,4104.765923076923)(53976.5,4275.635692307692)(54080.5,4462.988057692308)(54184.5,4669.346817307692)(54288.5,4896.447384615385)(54392.5,5141.418375)(54496.5,5409.600134615385)(54600.5,5705.087490384615)(54704.5,6036.476644230769)(54808.5,6409.3398942307695)(54912.5,6827.916634615385)(55016.5,7301.614615384615)(55120.5,7850.571971153846)(55224.5,8481.372740384615)(55328.5,9217.941913461538)(55432.5,10120.211701923075)(55536.5,11268.141509615383)(55640.5,12790.78876923077)(55697,13781.052)};
\addlegendentry{Prefix/suffix heuristic off}

\addplot[name path=pathZ3str3RE-ali , colourZ3str3RE-ali, line width=1.5pt,dash pattern={on 7pt off 2pt on 1pt off 3pt}] coordinates {(45501,1639.749)(45554.5,1644.165638888889)(45662.5,1652.987)(45770.5,1661.8768425925925)(45878.5,1670.8422592592594)(45986.5,1679.885138888889)(46094.5,1689.0165)(46202.5,1698.2432592592593)(46310.5,1707.5602592592593)(46418.5,1716.9974166666668)(46526.5,1726.5525277777776)(46634.5,1736.2058425925925)(46742.5,1745.9736944444444)(46850.5,1755.863888888889)(46958.5,1765.8876666666667)(47066.5,1776.0264444444445)(47174.5,1786.2887592592592)(47282.5,1796.697638888889)(47390.5,1807.2521666666667)(47498.5,1817.9380092592594)(47606.5,1828.7610555555555)(47714.5,1839.7176944444445)(47822.5,1850.8233055555554)(47930.5,1862.091425925926)(48038.5,1873.549)(48146.5,1885.18125)(48254.5,1896.9959537037037)(48362.5,1909.0166203703704)(48470.5,1921.2549444444444)(48578.5,1933.7511666666667)(48686.5,1946.5039814814816)(48794.5,1959.5097592592592)(48902.5,1972.7792592592593)(49010.5,1986.3572314814817)(49118.5,2000.232898148148)(49226.5,2014.3924814814816)(49334.5,2028.832037037037)(49442.5,2043.551462962963)(49550.5,2058.5151759259256)(49658.5,2073.6957592592594)(49766.5,2089.081351851852)(49874.5,2104.665527777778)(49982.5,2120.4500555555555)(50090.5,2136.460203703704)(50198.5,2152.714277777778)(50306.5,2169.2789166666666)(50414.5,2186.201814814815)(50522.5,2203.5089814814814)(50630.5,2221.20525)(50738.5,2239.3155555555554)(50846.5,2257.8193796296296)(50954.5,2276.6643055555555)(51062.5,2295.8929907407405)(51170.5,2315.5373981481484)(51278.5,2335.6155185185185)(51386.5,2356.2457407407405)(51494.5,2377.4712685185186)(51602.5,2399.454175925926)(51710.5,2422.2173055555554)(51818.5,2446.005388888889)(51926.5,2471.0940833333334)(52034.5,2497.9881944444446)(52142.5,2527.2879351851852)(52250.5,2559.095074074074)(52358.5,2592.402212962963)(52466.5,2626.796351851852)(52574.5,2662.113074074074)(52682.5,2698.1538240740742)(52790.5,2734.899296296296)(52898.5,2772.575490740741)(53006.5,2811.2731481481483)(53114.5,2851.151138888889)(53222.5,2892.6035555555554)(53330.5,2935.8565833333337)(53438.5,2981.45725)(53546.5,3030.17937037037)(53654.5,3082.2731018518516)(53762.5,3137.742638888889)(53870.5,3198.450777777778)(53978.5,3264.2998333333335)(54086.5,3334.1661296296297)(54194.5,3409.364)(54302.5,3489.725657407407)(54410.5,3576.094481481481)(54518.5,3670.6853240740743)(54626.5,3777.787527777778)(54734.5,3901.74425)(54842.5,4042.776685185185)(54950.5,4199.096018518519)(55058.5,4371.107111111111)(55166.5,4561.579592592592)(55274.5,4776.849916666667)(55382.5,5020.184037037037)(55490.5,5290.353888888889)(55598.5,5597.733351851852)(55706.5,5964.978888888889)(55814.5,6413.181027777778)(55922.5,6981.844259259259)(56080,8449.148)};
\addlegendentry{Automata length info off}

\addplot[name path=pathZ3str3RE-asi , colourZ3str3RE-asi, line width=1.5pt,dash dot] coordinates {(45501,1422.519)(45553.5,1426.5315)(45659.5,1434.5176132075471)(45765.5,1442.588141509434)(45871.5,1450.739924528302)(45977.5,1458.964)(46083.5,1467.2440283018868)(46189.5,1475.6035)(46295.5,1484.0260943396227)(46401.5,1492.524424528302)(46507.5,1501.1081320754718)(46613.5,1509.7836132075472)(46719.5,1518.5596792452832)(46825.5,1527.4311320754719)(46931.5,1536.4044811320755)(47037.5,1545.4937358490565)(47143.5,1554.7048396226414)(47249.5,1564.035858490566)(47355.5,1573.509990566038)(47461.5,1583.1437641509433)(47567.5,1592.9383113207548)(47673.5,1602.8964245283018)(47779.5,1613.009896226415)(47885.5,1623.290358490566)(47991.5,1633.6993301886794)(48097.5,1644.2416603773586)(48203.5,1654.9357169811321)(48309.5,1665.8348962264151)(48415.5,1676.9510849056605)(48521.5,1688.302716981132)(48627.5,1699.8801320754717)(48733.5,1711.7150471698112)(48839.5,1723.8248113207546)(48945.5,1736.2752641509435)(49051.5,1749.0884528301888)(49157.5,1762.257)(49263.5,1775.7873773584906)(49369.5,1789.6725471698114)(49475.5,1803.9058679245281)(49581.5,1818.4325094339622)(49687.5,1833.1511981132076)(49793.5,1848.029632075472)(49899.5,1863.0669716981133)(50005.5,1878.2939245283019)(50111.5,1893.6783773584907)(50217.5,1909.268896226415)(50323.5,1925.104141509434)(50429.5,1941.2131698113208)(50535.5,1957.6386981132075)(50641.5,1974.357858490566)(50747.5,1991.36570754717)(50853.5,2008.7121509433962)(50959.5,2026.4161037735848)(51065.5,2044.4131603773585)(51171.5,2062.6695)(51277.5,2081.2781132075474)(51383.5,2100.507971698113)(51489.5,2120.4823396226416)(51595.5,2141.6774905660377)(51701.5,2164.830349056604)(51807.5,2190.292962264151)(51913.5,2218.5847075471697)(52019.5,2249.263396226415)(52125.5,2280.866971698113)(52231.5,2313.184745283019)(52337.5,2346.15591509434)(52443.5,2379.660047169811)(52549.5,2413.6007264150944)(52655.5,2448.032603773585)(52761.5,2482.9427358490566)(52867.5,2518.3751226415093)(52973.5,2554.6441320754716)(53079.5,2592.356386792453)(53185.5,2632.086858490566)(53291.5,2673.937754716981)(53397.5,2718.4439716981133)(53503.5,2767.1118490566037)(53609.5,2818.6453679245283)(53715.5,2873.5613773584905)(53821.5,2933.7245566037736)(53927.5,2997.7934150943397)(54033.5,3065.624150943396)(54139.5,3138.889575471698)(54245.5,3215.960132075472)(54351.5,3295.827150943396)(54457.5,3381.3518113207547)(54563.5,3474.741830188679)(54669.5,3577.7845094339623)(54775.5,3692.3796226415093)(54881.5,3822.110650943396)(54987.5,3968.1155566037737)(55093.5,4130.356518867924)(55199.5,4316.7635)(55305.5,4532.412056603774)(55411.5,4781.1535)(55517.5,5080.276981132075)(55623.5,5464.794924528302)(55729.5,6034.521320754717)(55835.5,6960.829556603773)(55935,8394.022)};
\addlegendentry{Arithmetic solver integration off}

\addplot[name path=pathZ3str3RE , colourZ3str3RE, line width=1.5pt] coordinates {(45501,1132.544)(45555,1135.8187247706421)(45664,1142.3706972477064)(45773,1149.010119266055)(45882,1155.733)(45991,1162.5337155963302)(46100,1169.4090825688072)(46209,1176.366504587156)(46318,1183.407)(46427,1190.5347155963304)(46536,1197.7482844036697)(46645,1205.0439633027522)(46754,1212.415)(46863,1219.8613211009174)(46972,1227.3996788990826)(47081,1235.0264587155962)(47190,1242.7402568807338)(47299,1250.5410550458714)(47408,1258.447)(47517,1266.4467155963303)(47626,1274.5479266055045)(47735,1282.7485504587155)(47844,1291.0655688073396)(47953,1299.4913211009175)(48062,1308.0382568807338)(48171,1316.6953302752295)(48280,1325.4660550458716)(48389,1334.3503302752295)(48498,1343.3712844036697)(48607,1352.5324587155962)(48716,1361.842119266055)(48825,1371.3139724770642)(48934,1380.9709174311927)(49043,1390.813614678899)(49152,1400.86023853211)(49261,1411.1094495412844)(49370,1421.5664770642202)(49479,1432.2133302752295)(49588,1443.0682385321102)(49697,1454.1337706422019)(49806,1465.464256880734)(49915,1477.0873027522937)(50024,1489.0164311926605)(50133,1501.2539357798164)(50242,1513.9152293577981)(50351,1527.0297247706421)(50460,1540.5737064220184)(50569,1554.443486238532)(50678,1568.630486238532)(50787,1583.1049449541285)(50896,1597.7940275229357)(51005,1612.688128440367)(51114,1627.7702293577981)(51223,1643.0417889908256)(51332,1658.4941651376146)(51441,1674.124247706422)(51550,1689.9546788990826)(51659,1706.0075871559634)(51768,1722.3654678899081)(51877,1739.1150366972477)(51986,1756.2923119266054)(52095,1773.959614678899)(52204,1792.1975596330276)(52313,1811.1818899082568)(52422,1831.615614678899)(52531,1854.534623853211)(52640,1879.629770642202)(52749,1905.9822935779816)(52858,1933.3552018348623)(52967,1961.8611834862386)(53076,1991.5284220183487)(53185,2022.349880733945)(53294,2054.319798165138)(53403,2087.1910642201833)(53512,2121.040137614679)(53621,2156.039990825688)(53730,2192.310623853211)(53839,2230.9363761467894)(53948,2273.332376146789)(54057,2320.0670458715595)(54166,2370.8213486238533)(54275,2425.322990825688)(54384,2484.031807339449)(54493,2546.580990825688)(54602,2614.2907981651374)(54711,2688.8498990825688)(54820,2770.247642201835)(54929,2863.2391192660552)(55038,2972.2309357798163)(55147,3096.872550458716)(55256,3238.699889908257)(55365,3400.5507247706423)(55474,3590.5527981651376)(55583,3812.142091743119)(55692,4065.619559633027)(55801,4358.484302752294)(55910,4719.891614678899)(56019,5187.120256880734)(56159,6260.822)};
\addlegendentry{Z3str3RE (all heuristics on)}

\path[name path=axisZ3str3RE-none] (axis cs:0,0) -- (axis cs:53136,0);
\addplot [thick,color=colourZ3str3RE-none,fill=colourZ3str3RE-none,fill opacity=0.1] fill between [of=pathZ3str3RE-none and axisZ3str3RE-none];
\path[name path=axisZ3str3RE-li] (axis cs:0,0) -- (axis cs:53571,0);
\addplot [thick,color=colourZ3str3RE-li,fill=colourZ3str3RE-li,fill opacity=0.1] fill between [of=pathZ3str3RE-li and axisZ3str3RE-li];
\path[name path=axisZ3str3RE-psh] (axis cs:0,0) -- (axis cs:55697,0);
\addplot [thick,color=colourZ3str3RE-psh,fill=colourZ3str3RE-psh,fill opacity=0.1] fill between [of=pathZ3str3RE-psh and axisZ3str3RE-psh];
\path[name path=axisZ3str3RE-ali] (axis cs:0,0) -- (axis cs:56080,0);
\addplot [thick,color=colourZ3str3RE-ali,fill=colourZ3str3RE-ali,fill opacity=0.1] fill between [of=pathZ3str3RE-ali and axisZ3str3RE-ali];
\path[name path=axisZ3str3RE-asi] (axis cs:0,0) -- (axis cs:55935,0);
\addplot [thick,color=colourZ3str3RE-asi,fill=colourZ3str3RE-asi,fill opacity=0.1] fill between [of=pathZ3str3RE-asi and axisZ3str3RE-asi];
\path[name path=axisZ3str3RE] (axis cs:0,0) -- (axis cs:56159,0);
\addplot [thick,color=colourZ3str3RE,fill=colourZ3str3RE,fill opacity=0.1] fill between [of=pathZ3str3RE and axisZ3str3RE];
\end{axis}\end{tikzpicture}

}

\newcommand{\tableHeu}{
\begin{table}[t!]
\resizebox{1.0\textwidth}{!}{
\begin{tabular}{|c |c |c |c |c |c |c |}
\hline
&All off&\thead{Lazy intersection\\off}&\thead{Prefix/suffix\\off}&\thead{Automata length\\info off}&\thead{Arith. solver\\integ. off}&Z3str3RE\\ 
  \hline\hline 
sat &31046&31486&33817&33816&33804&\textbf{33820}\\ 
 \hline
unsat &22090&22085&21880&22264&22131&\textbf{22339}\\ 
 \hline
\hline 
 unknown &313&323&287&285&\textbf{283}&291\\ 
 \hline
timeout &3807&3362&1272&891&1038&\textbf{806}\\ 
 \hline
soundness error &\textbf{0}&\textbf{0}&\textbf{0}&\textbf{0}&\textbf{0}&\textbf{0}\\ 
 \hline
program crashes &42&39&\textbf{0}&1&\textbf{0}&\textbf{0}\\ 
 \hline
\hline 
 Total correct &53136&53571&55697&56080&55935&\textbf{56159}\\ 
 \hline
Time (s) &102102.388&101799.263&40068.501&27178.746&30006.857&\textbf{23339.266}\\ 
 \hline
Time w/o timeouts (s) &25962.388&34559.263&1462.8501&9358.746&9246.857&\textbf{7219.266}\\ 
 \hline
\end{tabular}}
\vspace{0.1cm}
 \caption{Comparison of different heuristics in \toolname{} on all benchmarks.}
    \label{tab:tableHeu}
\end{table}
}

\newcommand{\tableTotal}{
\begin{table}[t!]
\resizebox{0.95\textwidth}{!}{
\begin{tabular}{|c |c |c |c |c |c |c |}
\hline
&CVC4&Z3Seq&OSTRICH&Z3-Trau&Z3str3&Z3str3RE\\ 
  \hline\hline 
sat &33310&31550&22499&24133&27563&\textbf{33820}\\ 
 \hline
unsat &21897&21411&19281&21038&18566&\textbf{22339}\\ 
 \hline
\hline 
 unknown &\textbf{0}&\textbf{0}&10901&6504&1164&291\\ 
 \hline
timeout &2049&4295&4575&5581&9963&\textbf{806}\\ 
 \hline
soundness error &\textbf{0}&\textbf{0}&28&5325&13&\textbf{0}\\ 
 \hline
program crashes &\textbf{0}&\textbf{0}&\textbf{0}&2477&2&\textbf{0}\\ 
 \hline
\hline 
 Total correct &55207&52961&41752&39846&46116&\textbf{56159}\\ 
 \hline
Contribution score &95.99&19.87&--&--&--&\textbf{145.07}\\ 
\hline
Time (s) &57625.499&103487.844&305243.413&150288.386&213698.954&\textbf{23339.266}\\ 
 \hline
Time w/o timeouts (s) &16645.499&17587.844&213743.413&38668.386&14438.954&\textbf{7219.266}\\ 
 \hline
\end{tabular}}
\vspace{0.1cm}
 \caption{Combined results of string solvers on all benchmarks. \textbf{\toolname{}} has the best overall performance on all benchmarks compared to CVC4, OSTRICH, Z3seq, Z3str3, and Z3-trau and the biggest lead with a score of 1.02.}
    \label{tab:tableTotal}
\end{table}
}

%% file: abstract.tex
\begin{abstract}
  We present a novel length-aware solving algorithm for the quantifier-free first-order theory over regex membership predicate and linear arithmetic over string length. We implement and evaluate this algorithm and related heuristics in the Z3 theorem prover. A crucial insight that underpins our algorithm is that  real-world instances contain a wealth of information about upper and lower bounds on lengths of strings under constraints, and such information can be used very effectively to simplify operations on automata representing regular expressions. Additionally, we present a number of novel general heuristics, such as the prefix/suffix method, that can be used in conjunction with a variety of regex solving algorithms, making them more efficient. We showcase the power of our algorithm and heuristics via an extensive empirical evaluation over a large and diverse benchmark of \totalinstances{} regex-heavy instances, almost 75\% of which are derived from industrial applications or contributed by other solver developers. Our solver outperforms five other state-of-the-art string solvers, namely, CVC4, OSTRICH, Z3seq, Z3str3, and Z3-Trau, over this benchmark, in particular achieving a 2.4x speedup over CVC4, 4.4x speedup over Z3seq, 6.4x speedup over Z3-Trau, 9.1x speedup over Z3str3, and 13x speedup over OSTRICH. 
  
  \keywords{String solvers \and SMT solvers \and Regular expressions}
\end{abstract}



%% file: introduction.tex
\label{sec:intro}
Satisfiability Modulo Theories (SMT) solvers that support theories over regular expression (regex) membership predicate and linear arithmetic over length of strings, such as CVC4~\cite{CVC4-FROCOS15}, Z3str3~\cite{Z3str3}, Norn~\cite{norn}, S3P~\cite{S3P}, and HAMPI~\cite{hampi}, have enabled many important applications in the context of analysis of string-intensive programs. Examples include symbolic execution and path analysis~\cite{tacas09,willem}, as well as security analyzers that make use of string and regex constraints for input sanitization and validation~\cite{awszelkova,kaluza,jalangi}. 
Regular expression libraries in programming languages provide very intuitive and popular ways for developers to express input validation, sanitization, or pattern matching constraints.
Common to all these applications is the requirement for a rich quantifier-free (QF) first-order theory over strings, regexes, and integer arithmetic. Unfortunately, the QF first-order theory of strings containing regex constraints, linear integer arithmetic on string length, string-number conversion, and string concatenation (but no string equations\footnote{We use the terms ``word'' and ``string'' interchangeably in this paper.}) is undecidable (see Appendix). It can also be shown that many non-trivial fragments of this theory are hard to decide (e.g., they have exponential-space lower bounds or are PSPACE-complete). Therefore, the task of creating efficient solvers to handle practical string constraints that belong to fragments of this theory remains a very difficult challenge.

Many modern solvers typically handle regex constraints via an automata-based approach~\cite{abc-cav15}. Automata-based methods are powerful and intuitive, but solvers must handle two key practical challenges in this setting. The first challenge is that many automata operations, such as intersection, are computationally expensive, yet handling these operations is required in order to solve constraints that are relevant to real-world applications. The second challenge relates to the integration of length information with regex constraints. Length constraints derived from automata may imply a disjunction of linear constraints, which is often more challenging for solvers to handle than a conjunction. 

As we demonstrate in this paper, the challenges of using automata-based methods can be addressed via prudent use of \textit{lazy extraction of implied length constraints} and \textit{lazy regex heuristics} in order to avoid performing expensive automata operations when possible. Inspired by this observation, we introduce a length-aware automata-based algorithm, \toolname{} (and its implementation as part of the Z3 theorem prover~\cite{z3}), for solving regex constraints and linear integer arithmetic over length of string terms. \toolname{} takes advantage of the compactness of automata in representing regular expressions, while at the same time mitigating the effects of expensive automata operations such as intersection by leveraging length information and lazy heuristics. 

\smallskip
\noindent{\bf Contributions:} We make the following contributions in this paper.
\smallskip

\noindent \textbf{\toolname{}: An SMT Solver for Regular Expressions and Linear Integer Arithmetic over String Length.} In Section~\ref{sec:algorithm}, we present a novel decision procedure for the QF first-order theory over regex membership predicate and linear integer arithmetic over string length. We also describe its implementation, \toolname{}, as part of the Z3 theorem prover~\cite{Z3str3,z3}. The basic idea of our algorithm is that formulas obtained from practical applications have many implicit and explicit length constraints that can be used to reason efficiently about automata representing regexes. In Section~\ref{sec:heuristics} we present four heuristics that aid in solving regular expression constraints and that can be leveraged in general settings. Specifically, we present a heuristic to derive explicit length information directly from regexes, a heuristic to perform expensive automata operations lazily, a heuristic to refine lower and upper bounds on lengths of string terms with respect to regex constraints, and a prefix/suffix over-approximation heuristic to find empty intersections without constructing automata. All heuristics are designed to guide the search and avoid expensive automata operations whenever possible. Our solver, \toolname{}, handles the above theory as well as extensions (e.g. word equations and substring function) via the existing support in Z3str3. We focus on the core algorithm as it is the centerpiece of our regex solver. We also carefully distinguish the novelty of our method from previous work. 

\noindent \textbf{Empirical Evaluation and Comparison of \toolname{}\footnote{A reproduction package is available at \url{https://figshare.com/s/5ae73a6f3c55f5c5e4c1}} against CVC4, OSTRICH, Z3seq, Z3str3, and Z3-Trau:} To validate the practical efficacy of our algorithm, we present a thorough evaluation of \toolname{} in Section~\ref{sec:results} and compare it against CVC4~\cite{CVC4-CAV14}, OSTRICH~\cite{ostrich}, Z3's sequence solver~\cite{z3}, Z3str3~\cite{cav15}, and Z3-Trau~\cite{z3-trau} on \totalinstances{} instances across four regex-heavy benchmarks with connections to industrial security applications, including instances from Amazon Web Services and AutomatArk~\cite{automatark}. \toolname{} significantly outperforms other state-of-the-art tools on the benchmarks considered, having more correctly solved instances in total, lower running time, and fewer combined timeouts/unknowns than other tools, and no soundness errors or crashes. We note that almost 75\% of the benchmarks were obtained from industrial applications or other solver developers. Over all the benchmarks, we demonstrate a 2.4x speedup over CVC4, 4.4x speedup over Z3seq, 6.4x speedup over Z3-Trau, 9.1x speedup over Z3str3, and 13x speedup over OSTRICH. 

%% file: preliminaries.tex
\label{sec:prelim}

This section contains some basic definitions as well as a brief overview of the theoretical results which shape the landscape in which we state our contribution.

\subsection{Basic Definitions}
We first describe the syntax and semantics of the input language supported by our solver \toolname{} (Algorithm~\ref{alg:highlevel}).

\noindent{\bf Syntax:} The core algorithm we present in Section~\ref{sec:algorithm} accepts formulas of the quantifier-free many-sorted first-order theory of regex membership predicates over strings and linear integer arithmetic over string length function. The syntax of this theory is shown in Figure~\ref{fig:syntax}.

\begin{figure}[t!]
$\begin{array}{llll}
 F &\! \Coloneqq & Atom \vsep F \land F \vsep F \lor F \vsep \lnot F \\
 Atom &\! \Coloneqq & t_{str} \in RE \vsep A_{int} & \\ 
 A_{int} &\! \Coloneqq & t_{int} = t_{int} \vsep t_{int} < t_{int} \\
 RE &\! \Coloneqq & ``w" \!\! \vsep \!\! RE \cdot RE \!\!\vsep\!\! RE \cup RE \!\!\vsep\!\! RE^{*} \!\!\vsep\!\! \overline{RE}, \hspace{1mm}  \text{with} \hspace{1mm} w \in \mathrm{Con_{str}} & \\
 t_{int} &\! \Coloneqq & m \!\!\!\vsep\!\!\! v \!\!\!\vsep\!\!\! len(t_{str}) \!\!\!\vsep\!\!\! t_{int} + t_{int} \!\!\!\vsep\!\!\! m \cdot t_{int},   \text{with}  \hspace{1mm} m \in \mathrm{Con_{int}}, v \in Var_{int} \\
 t_{str} &\! \Coloneqq & s 
 , \hspace{1mm}   \text{with} \hspace{1mm} s \in \mathrm{Var_{str}} \cup \mathrm{Con_{str}}
  \end{array}$
\caption{Syntax of the input language accepted by Algorithm~\ref{alg:highlevel}. \toolname{} accepts an extension of this syntax supporting word equations and other string terms.}
\label{fig:syntax}
\end{figure}

We denote the set of all
string variables and all integer variables as $\mathrm{Var_{str}}$ and
$\mathrm{Var_{int}}$ respectively, and the set of all string constants and all
integer constants as $\mathrm{Con_{str}}$ and $\mathrm{Con_{int}}$ respectively.
String constants are any sequence of zero or more characters over a finite alphabet (e.g., ASCII).

Atomic formulas are regular expression
membership constraints and linear integer (in)equalities.
Regex terms are denoted recursively over
regex concatenation, union, Kleene star, and complement, and
for a string constant $w$, the regex term ``$w$'' represents the regular language
containing $w$ only.
All regex terms must be grounded (i.e. cannot
contain variables).  Linear integer arithmetic terms include integer
constants and variables, addition, and string length.
Multiplication by a constant is expanded to repeated addition.
String terms are either string variables or string constants.
The length of a string $S$ is denoted by $len(S)$, the number of characters in $S$.  The empty string has length 0.

Our implementation \toolname{} supports the theory in Figure~\ref{fig:syntax} extended with
more expressive functions and predicates, including word equations (equality between arbitrary string terms) 
and functions such as \texttt{indexof} and \texttt{substr} that are needed
for program analysis. \toolname{} handles these terms via existing support in Z3str3.
We focus on the above input language in the presentation of our algorithm in this paper
and theoretical content.

\smallskip

\noindent{\bf Semantics:} We refer the reader to~\cite{cav15} for a detailed description of the
semantics of standard terms in this theory. We focus here on the semantics
of terms which are less commonly known.
The regex membership predicate $S \in R$, where $S$ is a string term and $R$ is a regex term, is defined by structural recursion as follows:

$\begin{array}{lll}
  S \in ``w" & \mbox{iff} & \mbox{$S = w$ (where $w$ is a string constant) }\\
  S \in R_1 \cdot R_2 & \mbox{iff} & \mbox{there exist strings $S_1, S_2$ with $S = S_1 \cdot S_2$, $S_1 \in R_1$, $S_2 \in R_2$}\\
  S \in R_1 \cup R_2 & \mbox{iff} & \mbox{either $S \in R_1$ or $S \in R_2$}\\
  S \in R^{*} & \mbox{iff} & \mbox{either $S = \epsilon$ or there exists a positive integer $n$ such that}\\
  & & \mbox{$S = S_1 \cdot S_2 \cdot \hdots \cdot S_n$ and $S_i \in R$ for each $i = 1 \hdots n$} \\
  S \in \overline{R} & \mbox{iff} &  \mbox{$S \not\in R$ (that is, $S \in R$ is false)}
  \end{array}
$



\subsection{Theoretical Landscape} \label{sec:theory}

To put our contributions in context, we briefly discuss a series of (un)decidability and complexity results developed around the fragments and extensions of the theory supported by \toolname{}. 

In particular, we consider extensions which may have a string-number conversion predicate $numstr$~\footnote{We introduce $numstr$, which is not part of the SMT-LIB standard, in order to simplify presentation of the theoretical results. The predicate is no more expressive than the standard operators \texttt{str.to\_int}/\texttt{str.from\_int}, except that those terms handle decimal inputs. The results easily extend to other (finite) alphabets including decimal/hexadecimal digits with appropriate case analysis.} 
and/or string concatenation. Both extensions are important to real-world program analysis. The predicate
$numstr$ has the syntax $numstr(t_{int}, t_{str})$ and the following semantics: $numstr(n, s)$ is true for a given integer $n$ and string $s$ iff $s$ is a valid binary representation of the number $n$ (possibly with leading zeros) and $n$ is a non-negative integer. That is, $s$ only contains the characters 0 and 1, and $\sum_{i = 0}^{len(s) - 1} s'[i] 2^{len(s) - i - 1} = n$, where $s'[i]$ is 0 if the $i$th character in $s$ is `0' and 1 if that character is `1'. String concatenation has the syntax $t_{str} \Coloneqq  t_{str} \cdot t_{str}$ and the usual semantics defined by SMT-LIB~\cite{smtlib}.

In the following, $T_{LRE,n,c}$ is the quantifier-free many-sorted first-order theory of linear integer arithmetic over string length function ($L$), regex ($RE$) membership predicates, string-number conversion ($n$), and string concatenation~($c$)~\footnote{Note that the fragments considered here do not include word equations.}. The following quantifier-free fragments of $T_{LRE,n,c}$ are of interest: 
$T_{LRE,c}$, $T_{LRE}$, $T_{RE,n,c}$, $T_{RE,n}$, and $T_{RE}$. The fragment $T_{LRE,c}$ (respectively, $T_{LRE}$) has all functions and predicates of $T_{LRE,n,c}$ except the string-number conversion predicate (and, respectively, except the string concatenation function). 

The theory $T_{RE,n,c}$ (respectively, $T_{RE,n}$ and $T_{RE}$) has all functions and predicates of $T_{LRE,n,c}$ except the length function (and, respectively, the string concatenation function, and, in the case of $T_{RE}$, the string-number conversion predicate). Note that while all these theories allow equalities between terms of sort $Int$, they do not allow equalities between terms of sort $Str$, and therefore can not express general word equations.

The theoretical landscape is laid out as follows. Firstly, following the results and techniques introduced in \cite{norn}, we obtain that $T_{LRE,c}$ and, in particular, $T_{LRE}$ is decidable. A procedure deciding a formula from $T_{LRE,c}$ would first construct for each variable (string or integer), based on the regular expression constraints and length constraints which involve it, a finite automaton, then reduce the problem of checking the satisfiability of the formula to checking whether the constructed automata accept at least one string. A similar approach shows that $T_{RE,n}$ is decidable.

We observe that the presence of complements in regular expressions is an inherent source of complexity for these procedures. Indeed, we can easily encode the universality problem for regular expressions as a formula in the theory $T_{RE}$. Moreover, given a regex $R$ of length $n$ over an alphabet $\Sigma$, deciding whether $L(R)=\Sigma^*$ is equivalent to deciding the satisfiability of the formula $\varphi$ of $T_{RE}$ consisting of the atoms $x\in \overline{R}$ and $x\in \Sigma^*$. Accordingly, by the results from~\cite{stockmeyer}, if the choice for $R$ is restricted to regular expressions with at least $k$ stacked complements, then there exists a positive rational number $c$ such that the considered problems are not contained in NSPACE$\left({\underbrace{{2
  {{{^{2\vphantom{h}}}^{2\vphantom{h}}}^{\cdots\vphantom{h}}}^{2\vphantom{h}}}
}_{\text{$k-1$ times}}}^{cn}\right)$. In other words, the depth of the stack of complements of the formula translates to the height of the tower of exponents in the complexity of deciding that formula $\varphi$. On the other hand, if we only consider regular expressions without stacked complements, then the decision problems for the considered theories are PSPACE-complete. Indeed, the automata-based approach described above can be implemented to work in nondeterministic polynomial space; strongly related complexity results are obtained in \cite{LinB16,LinM18}. 

At the opposite end of the spectrum is the theory $T_{LRE,n,c}$, which is undecidable. Indeed, one can show that the more specific theory $T_{RE,n,c}$ (i.e. disallowing arithmetic over length) has equivalent expressive power to the theory of word equations with regular constraints, a predicate allowing the comparison of the length of string terms, and the $numstr$ predicate. Therefore, using the techniques from~\cite{rp2018strings}, one can show that the theory $T_{LRE,n,c}$, in which we additionally allow arithmetic over length, is undecidable. (See Theorem~\ref{thm:undecWE} in the Appendix.)

The inherent extremely high complexity of the satisfiability problems from $T_{RE}$ and of the theories extending it, as well as the undecidability of the respective problem for $T_{LRE,n,c}$, suggest that the usage of heuristics will decisively influence any algorithmic approaches to solving these problems in practice.

%% file: algorithm-2020.tex
\label{sec:algorithm}


This section outlines the high-level algorithm used by \toolname{} to solve the satisfiability problem for $T_{LRE}$, and its extension based on length-aware heuristics. \looseness=-1

\subsection{High-Level Algorithm}

\begin{algorithm}[t!]
  \scriptsize
  \caption{The length-aware algorithm for the theory $T_{LRE}$ of regex and integer constraints} \label{alg:highlevel}
  \SetKwFunction{ComputeLengthAbstraction}{ComputeLengthAbstraction}
  
  \Input{Conjunction $\phi$ of constraints of the form $S \in RE$, and conjunction $\psi$ of linear integer arithmetic constraints}
  \Output{SAT or UNSAT}
  
  \ForAll{constraints $S \in RE$ in $\phi$}{ \label{line:lengthabs1}
    $L_{S} \gets$ \ComputeLengthAbstraction{$S$} \;
    $L_{RE} \gets$ \ComputeLengthAbstraction{$RE$} \; \label{line:regexlengthabstraction}
    \If{$\psi \cup L_{S} \cup L_{RE}$ inconsistent}{ \label{line:regexinconsistent}
      \Return{UNSAT} \label{line:regexreturnunsat}
    }
    refine $L_{S}$ as tightly as possible with respect to $L_{RE}$\; \label{line:refinelength}
  } \label{line:lengthabs2} 
  
  \ForAll{strings $S_i$ occurring in $\phi$}{ \label{line:intersect1}
    let $\mathcal{R}$ be the set of all regexes $RE$ in all terms $S_i \in RE$ \;
    
    $I \gets$ intersection of all regular expressions in $\mathcal{R}$ \; \label{line:intersect}
    \If{$I$ is empty}{
      \Return{UNSAT}
    }
  } \label{line:intersect2} 
  
  $\mathcal{L}_{S} \gets$ the union of all length abstractions $L_{S}$\;
  $\mathcal{L}_{RE} \gets$ the union of all length abstractions $L_{RE}$\;
  
  \ForAll{solutions $M$ of the system $\psi \cup \mathcal{L}_S \cup \mathcal{L}_{RE}$}{ \label{line:alllengths}
    \ForAll{strings $S$ occurring in $\phi$}{ \label{line:lengthsearchstart}
      $len(S) \gets M[len(S)]$ \;
      let $\mathcal{A}$ be the set of all automata for all regexes $RE$ in all terms $S_i \in RE$ \;
      determinize all accepting paths of length $len(S)$ for each $A$ in $\mathcal{A}$ \;
      \eIf{each automaton has an accepting path of length $len(S)$ and all paths are character-consistent}{ \label{line:characterconsistent}
        continue
      }{
        stop processing and try the next solution $M$
      }
    } 
    \If{all strings $S$ were processed successfully with respect to this $M$}{
      \Return{SAT}
    } \label{line:lengthsearchend}
  } 
  
  \Return{UNSAT} \label{line:nosolutions}
  
\end{algorithm}

The pseudocode presented in Algorithm~\ref{alg:highlevel} is shown at a high level
that captures the essence of the procedure being performed. 
Implementation-specific details are omitted for clarity.
\toolname{} incorporates a version of this algorithm as part of a
DPLL(T)-style interaction with a core solver for Boolean combinations of atoms
and other theory solvers able to handle arithmetic constraints and other terms.
The tool handles string concatenation, string equality,
and other string terms and predicates besides regex membership and string length via existing support in Z3str3,
and leverages Z3's integer arithmetic solver for arithmetic reasoning and model construction.
This high-level presentation is expanded in Section~\ref{sec:heuristics}, where we describe several heuristics used in our implementation as part of the \toolname{} tool.

The algorithm takes as input a conjunction $\phi$ of regex membership constraints
and a conjunction $\psi$ of linear integer arithmetic constraints over the lengths of string variables appearing in $\phi$.
Without loss of generality, it is assumed that all constraints in $\phi$ are positive; negative
constraints $S \not\in RE$ can be replaced with the positive complement $S \in \overline{RE}$.
The algorithm returns SAT if there is a satisfying assignment to all string variables consistent with the regex constraints $\phi$ and length constraints $\psi$.
It is assumed that the algorithm has access to a procedure for checking the consistency of
linear integer arithmetic constraints and for obtaining satisfying assignments to these constraints
(in our implementation, this is fulfilled by Z3's arithmetic solver).

Lines~\ref{line:lengthabs1}--\ref{line:lengthabs2} check whether the length information implied by $\phi$ is consistent with $\psi$.
The function \texttt{ComputeLengthAbstraction} takes as input either a string term $S$ or a regex
$RE$ and computes a system of length constraints corresponding to either an abstraction of derived length information from string constraints
or an abstraction of length information derived from the regex $RE$. 
For example, given the regex $(abc)^{*}$ as input, \texttt{ComputeLengthAbstraction} would construct the length abstraction $S \in (abc)^{*} \to len(S) = 3n, n \ge 0$ for a fresh integer variable $n$.
If the length abstractions are inconsistent with the given length constraints, there can be no
solution which satisfies both the length and regex constraints, and hence the algorithm returns UNSAT. Otherwise, line~\ref{line:refinelength} refines the length abstraction $L_S$ with respect to the regex $RE$. This improves the efficiency of finding solutions to the augmented system of length constraints later in the algorithm. In our implementation, the lower and upper bounds of the length of $S$ are checked against the lengths of accepting paths in the automaton for $RE$.
For instance, if $L_S$ implies that $len(S) \ge 5$, but the shortest accepting path
in the automaton has length 7, the lower bound is refined to $len(S) \ge 7$.

Lines \ref{line:intersect1}--\ref{line:intersect2} check that the intersection of all automata constraining each string variable
is non-empty. Although intersecting automata is relatively expensive (as it runs in quadratic time w.r.t. the size of the intersected automata), it is still more efficient to do this before enumerating length assignments, and taking the intersection here is necessary to maintain soundness.
(The heuristics in Section~\ref{sec:heuristics} illustrate some methods by which this computation can be made more efficient or even avoided.)

At this point in the algorithm, the length constraints $\psi$ (as well as our length abstractions)
are consistent, and the regex constraints $\phi$ are consistent, so we can
check the joint consistency of $\phi \land \psi$.
The algorithm enumerates each possible model $M$ of our augmented system of length constraints
(line~\ref{line:alllengths}), obtaining
a candidate assignment for the length of each string in $\phi$.
It is then necessary to check whether solutions of this length actually exist for each
regex constraint in $\phi$.
The algorithm does this by enumerating all accepting paths in the corresponding automata
having that respective length.
Line~\ref{line:characterconsistent} checks the character-consistency of each accepting path
of length $len(S)$ in each automaton $A$. Here, by ``character-consistent'', we mean that some accepting path in each automaton follows transitions corresponding to the same characters of $S$.
This can be checked in various ways, as the total number of paths, as well as their size, is always finite. For example, in our implementation of this algorithm in \toolname{}, the solver converts each path to a disjunction of bit-vector character constraints and checks the satisfiability of the resulting system. If each string $S$ is mapped to a path through the automata corresponding to its regex constraints, of a length that is consistent with the arithmetic constraints, then the system is satisfiable and the algorithm returns SAT. Otherwise, this process repeats for the next length assignment $M$.
Line~\ref{line:nosolutions} is reached if and only if no solution of the combined system of length constraints has a satisfying assignment over string variables with that length. If this happens, the constraints $\phi$ and $\psi$ are not jointly satisfiable and the algorithm returns UNSAT. 

We demonstrate soundness, completeness, and termination of Algorithm~\ref{alg:highlevel} as follows. On line~\ref{line:regexinconsistent} we check whether $\psi \cup L_{S} \cup L_{RE}$ is satisfiable. If not, we return UNSAT on line~\ref{line:regexreturnunsat}. Therefore, if control reaches line~\ref{line:alllengths}, $\psi \cup L_{S} \cup L_{RE}$ must be satisfiable. By construction of $L_{RE}$, every solution to $\phi$ is also a solution to $L_{RE}$ and vice versa (all solutions are preserved). It follows that satisfiability of $\psi \cup L_{S} \cup L_{RE}$ at line~\ref{line:alllengths} implies the existence of a smallest valid string solution that satisfies all length constraints. The loop body (lines~\ref{line:lengthsearchstart}--\ref{line:lengthsearchend}) finds and outputs this smallest valid string solution.
Therefore, Algorithm~\ref{alg:highlevel} is a decision procedure for the QF first-order theory of regex constraints, string length, and linear integer arithmetic. 

The conditionals inside the loop body (lines~\ref{line:lengthsearchstart}--\ref{line:lengthsearchend})  suggest that we may discover unsatisfiability. As previously mentioned, \toolname{} supports other high-level operations that are not part of this theory via existing support in Z3str3. These conditionals provide support for including these operations, which may render the theory undecidable.
These terms are not in Algorithm~\ref{alg:highlevel} because their inclusion would make the algorithm incomplete (see Section~\ref{sec:theory}).
Algorithm~\ref{alg:highlevel} describes the part of the implementation which is novel and complete.

\section{Length-Aware and Prefix/Suffix Heuristics in \toolname{}} \label{sec:heuristics}

In this section, we describe the length-aware heuristics that are used in \toolname{}
to improve the efficiency of regular expression reasoning.
We present an empirical evaluation of the power of these heuristics in Section~\ref{sec:empiricalheuristics}.

\subsection{Computing Length Information from Regexes}

The first length-aware heuristic we use is when constructing the length abstraction on line~\ref{line:regexlengthabstraction}. If the regex can be easily converted to a system of equations describing the lengths of all possible solutions (for instance, in the case when it does not contain any complements or intersections), this system can be returned as the abstraction without constructing the automaton for $RE$ yet. 
As previously illustrated, for example, given the regex $(abc)^{*}$ as input, \texttt{ComputeLengthAbstraction} would construct the length abstraction $S \in (abc)^{*} \to len(S) = 3n, n \ge 0$ for a fresh integer variable $n$.
Note that this can be done from the syntax of the regex without converting it to an automaton. Deriving length information from the automaton would be simple by, for example, constructing a corresponding unary automaton and converting to Chrobak normal form. However, performing automata construction lazily means we cannot rely on having an automaton in all cases; this technique also provides length information even when constructing an automaton would be expensive.

In cases where we cannot directly infer the length abstraction,
the heuristic will fix a lower bound on the length of words in $RE$, and possibly an upper bound if it exists. Reasoning about the length abstraction early in the procedure gives our algorithm
the opportunity to detect inconsistencies before expensive automaton operations are performed.
This gives the arithmetic solver more opportunities to propagate facts discovered by refinement
and potentially more chances to find inconsistencies or learn further derived facts.

\subsection{Optimizing Automata Operations via Length Information}

Similarly, computing the intersection $I$ in line~\ref{line:intersect} is done lazily in the implementation of \toolname{} and over several iterations of the algorithm.
The most expensive intersection operations can be performed at the end of the search, after as much other information as possible has been learned. We use the following heuristics recursively to estimate the ``cost'' of each operation without actually constructing any automata:
\begin{itemize}
    \item For a string constant, the estimated cost is the length of the string.
    \item For a concatenation or a union of two regex terms $X$ and $Y$, the estimated cost is the sum of the estimates for $X$ and $Y$.
    \item For a regex term $X^{*}$, the estimated cost is twice the estimate for $X$.
    \item For a regex term $X$ under complement, the estimated cost is the product of the estimates obtained from subterms of $X$.
\end{itemize}

In essence, the constructions which ``blow up'' the least are expected to be the least expensive and are performed first. In the 
best-case scenario, this could mean avoiding the most expensive operations completely if an intersection of smaller automata ends up being empty. In the worst case, all intersections are computed eventually, as this is necessary to maintain the soundness of our approach.

\subsection{Leveraging Length Information to Optimize Search}

Our implementation communicates integer assignments and lower/upper bounds with the external
arithmetic solver in order to prune the search space. The search for length assignments on line~\ref{line:alllengths} is done in practice as an abstraction-refinement loop involving Z3's arithmetic solver. The arithmetic solver proposes a single candidate model for the system of arithmetic constraints; the regex algorithm checks whether that model has a corresponding solution over the  regex constraints. If it does not, it asserts a conflict clause blocking that combination of length assignments and regex constraints from being considered again. This is necessary in a DPLL(T)-style solver such as Z3 in order to handle Boolean structure in the input formula.

\subsection{Constructing Over-Approximated Prefixes/Suffixes to Find Empty Intersections} 

As previously mentioned, computing automata intersections is expensive, but in many cases it is necessary in order to prove that a set of intersecting regex constraints has no solution. In some cases, this can be done ``by inspection'' from the syntax of the regex terms without constructing or intersecting any automata. From the structure of a regular expression, it is easy to determine the first letter of all possible accepted strings that it matches. If several regexes would be intersected over the same string term, this is used to check whether these regexes have a prefix of length one in common. If they do not, their intersection cannot contain any strings other than the empty string (and we can also check whether the empty string could be accepted by a similar syntactic approach). 
A similar construction for suffixes of length 1 is also used.
In this way, the heuristic can infer that the intersection of several regex constraints is either empty, resulting in a conflict clause, or can only contain the empty string, resulting in a new fact and a simplification of the formula -- without actually constructing the intersection or, in fact, constructing any automata for these regexes.

For example, consider the following regex constraints on a variable $X$:
\begin{align*}
    X & \in (abc)^{*} \\
    X & \in a^{+}\;|\;b^{+}
\end{align*}
The prefix/suffix heuristic would proceed as follows. In the first constraint,
the pattern $abc$ is matched zero or more times, and could be empty; therefore,
either $X$ is empty or it must start with $a$ and end with $c$.
In the second constraint, each pattern is matched at least once, and cannot be empty;
therefore $X$ must start with $a$ or $b$, end with $a$ or $b$,
and cannot be the empty string.
Observe that according to the prefix heuristic, these constraints are consistent,
since $a$ is a valid prefix of both regexes;
however, according to the suffix heuristic, they are inconsistent,
as the possible suffixes $a$ and $b$ of the second regex do not include $c$,
and the empty string is not a solution to both constraints.
Hence we conclude that these constraints are not jointly satisfiable,
and assert a conflict clause.

As demonstrated, all of these facts are derived from the syntax of the regular expression;
the heuristic does not need to construct an automaton for either constraint
in order to reason about them.
By constructing an over-approximation of the possible solutions of $X$ allowed by regex constraints,
the heuristic can determine that their intersection is empty
(or can only contain the empty string) without computing it precisely
(which, as previously mentioned, is expensive to do and also requires constructing automata first).
We limit this heuristic to the first letter as each additional letter causes the space required to keep track of these prefixes to increase exponentially.

%% file: experimental-results.tex
\begin{figure}[t!]
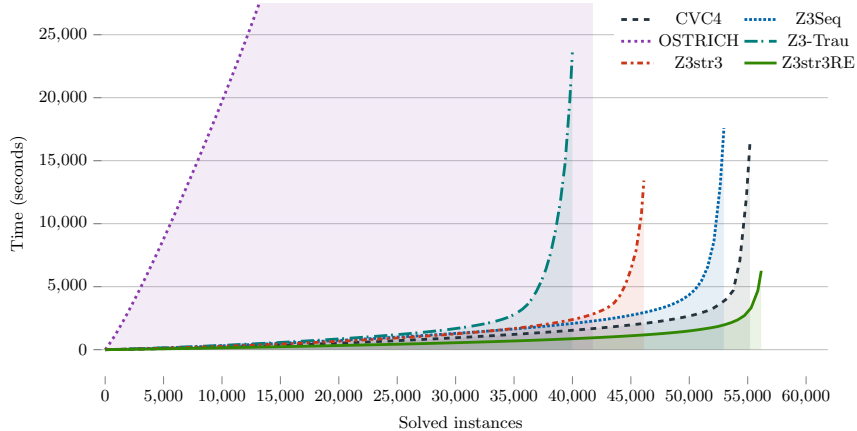

\begin{center}
\resizebox{.95\textwidth}{!}{\cactusTotalCroped{}}
\end{center}

\vspace{-0.75cm}

\caption{Cactus plot summarizing performance on all benchmarks. \toolname{} has the best overall performance.\label{fig:catctus_total}} 
\end{figure}

\tableTotal{}

\label{sec:results}

In this section, we describe the empirical evaluation of \toolname{}, our implementation of the
length-aware regular expression algorithm
presented in Section~\ref{sec:algorithm}, to validate the effectiveness of the techniques presented.
We evaluate the correctness and efficiency of our tool against other solvers,
as well as against different configurations of the tool in order to demonstrate
the efficacy of the heuristics we describe.

\subsection{Empirical Setup and Solvers Used}

We compare \toolname{} 
against five other leading string solvers available today.
CVC4~\cite{CVC4-CAV14} is a
general-purpose SMT solver which reasons about strings and regular
expressions algebraically.  Z3str3~\cite{Z3str3} is the latest solver
in the Z3-str family, and uses a reduction to word equations to reason
about regular expressions.
\toolname{} is based on Z3str3 except for the length-aware algorithm and heuristics described in Sections~\ref{sec:algorithm} and~\ref{sec:heuristics}.
Z3seq~\cite{stanford2020symbolic} is the Z3 sequence solver,
implemented by Nikolaj Bj\o{}rner and others at Microsoft Research,
as part of the Z3 theorem prover. Z3seq uses a new theory of derivatives for solving extended regular expressions.
Z3-Trau~\cite{z3-trau} is also based on Z3 and uses an automata-based approach known as ``flat automata'' with both under- and over-approximations.
OSTRICH~\cite{ostrich} uses a reduction from string functions (including word equations) to a model-checking problem that is solved using the SLOTH tool and an implementation of IC3.
We used CVC4's binary version 1.8,
commit \texttt{59e9c87} 
of Z3str3, the sequence solver included in Z3's binary version 4.8.9, Z3-Trau commit \texttt{1628747}, and OSTRICH version 1.0.1.
All of these tools support the full SMT-LIB standard for strings.
We did not compare against the Z3str2~\cite{cav15} or Norn~\cite{norn} solvers as neither tool
supports the \texttt{str.to\_int} or \texttt{str.from\_int} terms which
represent string-number conversion, which are used in some sanitizer
benchmarks.  Additionally, Norn does not support many of the other
high-level string terms such as \texttt{indexof} or \texttt{substr} which
are used in the benchmarks. The ABC~\cite{abc-cav15} solver handles string
and length constraints by conversion to automata. However, their method over-approximates the solution set of the input formula which may be unsound. Thus, we excluded ABC from our
evaluation. We also were unable to evaluate against Trau~\cite{AbdullaFlattenAndConquer} as the provided source code did not compile.
All evaluations were performed on a
server running  Ubuntu 18.04.4 LTS with two AMD EPYC 7742 processors and 2TB of memory using the \text{ZaligVinder}~\cite{zaligvinder} benchmarking framework. 
A 20 second timeout was used. We cross-verified the models generated by each solver
for satisfiable instances against all competing solvers.

\input{data1.tex}

\subsection{Benchmarks}\label{sec:benchmarks}
The comparison was performed on four suites of regex-based benchmarks with a total of \totalinstances{} instances. In total, almost 75\% of the instances in our evaluation came from previously published industrial benchmarks or other solver developers.
Over all instances, 10\% contain extended regular expressions (having either complement or intersection, or both).
We briefly describe each benchmark's origin and composition.

\smallskip

\noindent \textbf{\benchmarkregexlib{}} is a set of 19979 benchmarks based on a collection of real-world regex queries collected by Loris D'Antoni from the University of Wisconsin, Madison, USA. We translated the provided regexes~\cite{automatark} into SMT-LIB syntax resulting in two sets of instances: a ``simple'' set with a single regex membership predicate per instance, and a ``complex'' set with 2--5 regex membership predicates (possibly negated) over a single variable per instance. The instances in this benchmark are evenly divided between simple and complex problems.

\smallskip

\noindent \textbf{\benchmarkcollected{}} is a set of 22425 instances taken from existing benchmarks with the purpose of evaluating the performance of solvers against real-world regex instances. This benchmark includes all instances from the AppScan~\cite{Z3str2-FMSD}, BanditFuzz,\footnote{The BanditFuzz benchmark is an unpublished suite obtained via private communication with the authors.} JOACO~\cite{thome2018integrated}, Kaluza~\cite{kaluza}, Norn~\cite{norn}, Sloth~\cite{sloth}, Stranger~\cite{yu2010}, and Z3str3-regression~\cite{Z3str3} benchmarks in which at least one regex membership constraint appears.\footnote{Other benchmark suites available to us, including the PyEx, PISA, and Kausler benchmarks, did not include any regex membership constraints.}
No additional restrictions are placed on which instances were chosen besides the presence of at least one regex membership predicate.  We chose to evaluate against this benchmark in order to test the performance of solvers against instances that are already known to be challenging to solve and that have appeared in previously published and widely distributed benchmark suites. 
Additionally, these instances may contain regex terms in any context and with any other supported string operators. As a result, the benchmark is also exemplary of how string solvers perform in the presence of operations and predicates that are relevant to program analysis.

\smallskip

\noindent \textbf{\benchmarkgen{}} is a set of 4170 problems generated by the StringFuzz string instance fuzzing tool~\cite{stringfuzz}. These instances only contain regular expression and linear arithmetic constraints. The motivation in choosing this benchmark is to isolate and evaluate the regex performance of a string solver in the context of mixed regex and arithmetic constraints. Tools with better regex and arithmetic solvers should perform better.
Fuzz testing, as performed in the \textbf{\benchmarkgen{}} benchmark, has been shown to be extremely productive in discovering bugs and performance issues in SMT solvers. We chose to include these instances because they enable us to isolate the performance of the solver on regex-heavy constraints in a way that the industrial benchmarks or instances obtained from other solver developers cannot.

\smallskip

\input{data2.tex}

\noindent \textbf{\benchmarkxform{}} is a set of 10682 instances which were produced by transforming existing industrial instances with StringFuzz.
To create the \textbf{\benchmarkxform{}} benchmark, we applied StringFuzz's transformers to instances supplied by Amazon Web Services related to security policy validation, handcrafted instances which are inspired by real-world input validation vulnerabilities, and the regex test cases
included in Z3str3's regression test suite.
The instances in this suite include regex constraints,
arithmetic constraints on string length, string-number conversion
($numstr$), 
string concatenation, word equations, and other high-level string operations such as \texttt{charAt}, \texttt{indexof}, and \texttt{substr}. 
As is typical for fuzzing in software testing, the goal is to create a
suite of tests from a given input that are similar in structure but
that explore interesting behaviour not captured by a ``typical''
industrial instance. These transformed instances are often harder than the original industrial ones. 

\input{data3.tex}

\subsection{Comparison and Scoring Methods}

We compare solvers directly against the total number of correctly solved cases, total time with and without timeouts, and total number of soundness errors and program crashes.
We also computed the biggest lead winner and largest contribution ranking following the scoring system used by the SMT Competition~\cite{smtcomprules2020}. Briefly, the biggest lead measures the proportion of correct answers of the leading tool to correct answers of the next ranking tool, and the contribution score measures what proportion of instances were solved the fastest by that solver. In accordance with the SMT Competition guidelines, a solver receives no contribution score (denoted as --) if it produces any incorrect answers on a given benchmark. In both cases, higher scores are better.

\subsection{Analysis of Empirical Results}

The cactus plot in Figure \ref{fig:catctus_total} shows the cumulative time taken by each solver on all cases in increasing order of runtime. Solvers that are further to the right and closer to the bottom of the plot have better performance.

Overall \toolname{} solves more instances and performs better than all competing solvers. Across all benchmarks, \toolname{} is over 2.4x faster than CVC4, 4.4x faster than Z3seq, 6.4x faster than Z3-Trau, 9.1x faster than Z3str3, and 13x faster than OSTRICH
(including timeouts). Additionally, \toolname{} has fewer combined timeouts and unknowns than other tools considered, and no soundness errors or crashes.  We summarize these results in Table \ref{tab:tableTotal}.
Notably, both Z3-Trau~\cite{z3-trau} and OSTRICH~\cite{ostrich} had significant runtime issues in our experiments. Z3-Trau produced 5325 soundness errors and 2477 crashes on our benchmarks (13\% of all instances), which is significantly higher than other tools used. OSTRICH produced 10901 ``unknown'' responses on the benchmarks (19\% of all instances), due to both unsupported features and crashes, and also produced 28 soundness errors.
Over all benchmarks, Z3str3RE produced 291 unknowns. There are several potential reasons for this; the solver may have encountered a resource limit and returned UNKNOWN, or Z3str3 may have detected non-termination and returned UNKNOWN instead of looping forever.
According to SMT Competition scoring, \toolname{} won the division across all benchmarks with a lead of 1.02, and had the largest contribution to the division with a score of 145.07. CVC4 had a contribution score of 95.99, and Z3seq had a score of 19.87. OSTRICH, Z3-Trau, and Z3str3 received no contribution score as they each returned at least one incorrect answer.
The presented results are typical of the performance of the evaluated tools over multiple runs. Results were cross-validated within runs and between multiple runs. For a random single instance, the sample variance in execution time for 100 runs is 0.001 (0.07\% of average execution time). Over \totalinstances{} instances, this is negligible.

The empirical results make clear the
efficacy of length-aware automata-based techniques for regular
expression constraints when accompanied with length constraints (which
is typical for industrial instances). The effectiveness of our technique is demonstrated particularly by comparing \toolname{} with Z3str3, as the only differences between these tools are the length-aware regex algorithm and heuristics implemented in \toolname{} and bug fixes.
By improving the regex algorithm and applying our heuristics,
we achieved a speedup of over 9x and solved over 10000 more cases than Z3str3.

\input{data4.tex}

\subsection{Detailed Experimental Results}

Figure~\ref{fig:cactus_automatark} and Table~\ref{tab:cactus_automatark} show the 
detailed results for the \textbf{\benchmarkregexlib{}} benchmark.
In this benchmark, \toolname{} solves more instances than all other solvers,
has the fewest timeouts/unknowns,
and has the fastest overall running time.
Including timeouts, \toolname{} is 2.2x faster than CVC4, 4.7x faster than Z3seq,
40.4x faster than OSTRICH, 20.4x faster than Z3-Trau,
and 32.3x faster than Z3str3.

Figure~\ref{fig:cactus_generated} and Table~\ref{tab:cactus_generated} show the 
detailed results for the \textbf{\benchmarkgen{}} benchmark.
\toolname{} solves more instances than all other solvers, has over 90\% fewer
timeouts than other solvers, no unknowns, and has the fastest overall running time.
Including timeouts, \toolname{} is 6.1x faster than CVC4,
6.9x faster than Z3seq, 10x faster than OSTRICH, 7.3x faster than Z3-Trau,
and 4.3x faster than Z3str3.

Figure~\ref{fig:cactus_transformed} and Table~\ref{tab:cactus_transformed} show the 
detailed results for the \textbf{\benchmarkxform{}} benchmark.
\toolname{} solves more instances in total than all other solvers and has
the lowest total running time without timeouts.
Including timeouts, \toolname{} is 2.7x faster than CVC4,
1.9x faster than Z3seq, 21x faster than OSTRICH,
and 27x faster than Z3str3.
Although Z3-Trau is 1.5x faster than \toolname{} on this benchmark, including timeouts,
Z3-Trau also produces 1241 answers with soundness errors and crashes on 718 other cases.
\toolname{} produces no wrong answers or soundness errors on the benchmark.
Z3-Trau also solves 1923 fewer cases correctly in total than \toolname{}.

Figure~\ref{fig:cactus_collected} and Table~\ref{tab:cactus_collected} show the 
detailed results for the \textbf{\benchmarkcollected{}} benchmark.
\toolname{} outperforms Z3seq, Z3str3, OSTRICH, and Z3-Trau on this benchmark and is competitive
with CVC4 both in terms of total number of instances correctly solved
and total running time.
CVC4 solves 609 more instances than \toolname{} on this benchmark, but \toolname{}
is 1.1x faster overall (including timeouts). 
\toolname{} is 3.6x faster than Z3seq, 5.4x faster than OSTRICH,
2.4x faster than Z3-Trau, and 2.6x faster than Z3str3.

\begin{figure}[t!]
    \centering
\begin{resizebox}{.95\textwidth}{!}
    \cactusHeuristics
\end{resizebox}

\caption{Cactus plot comparing performance by disabling individual heuristics on all benchmarks.\label{fig:cactus_heu}} 
\end{figure}

\tableHeu{}

\subsection{Analysis of Individual Heuristics and Results} \label{sec:empiricalheuristics}

To demonstrate the effectiveness of individual heuristics described in Section~\ref{sec:heuristics}
and implemented in \toolname{}, we evaluated different configurations of the tool in which one or more heuristics were disabled.
Figure~\ref{fig:cactus_heu} and Table~\ref{tab:tableHeu} show the results.
The plot line ``Z3str3RE'' shows
the baseline performance of the tool with all heuristics enabled. The plot line ``All heuristics off'' shows the performance with all heuristics disabled.
Each of the other plot lines shows the performance with the named heuristic disabled and all others kept enabled.
From the plots and table, it is clear that \toolname{} performs best with all heuristics enabled. \toolname{} is 4.4x faster
using all our heuristics than using none. Every other configuration of the tool performs significantly worse relative to the one with all heuristics enabled. Also, the length-aware and prefix/suffix heuristics provide significant boost over lazy intersections and the baseline. These results demonstrate empirically that each heuristic we introduce provides significant benefit in both total number of solved instances and total solver runtime, and that all of the heuristics can be used simultaneously for maximum efficacy.

%% file: data1.tex
\begin{figure}[t!]

\resizebox{.95\textwidth}{!}{\pgfplotsset{scaled x ticks=false}\pgfplotsset{scaled y ticks=false}\begin{tikzpicture}\begin{axis}[title=Automatark,xmin=-1000,xlabel=Solved instances,ylabel=Time (seconds),,legend columns=2,legend style={at={(1,0.95)},nodes={scale=1, transform shape}, fill=none,anchor=east,align=center },axis line style={draw=none}, xtick pos=left, ytick pos=left, ymajorgrids=true, legend style={draw=none,fill=white},x post scale=2,y post scale=1.25,ymax=12500]
\addplot[name path=pathCVC4 , colourCVC4, line width=1.25pt,dashed] coordinates {(1,0.008)(101,1.1842338308457712)(302,4.144800995024875)(503,7.587910447761193)(704,11.333179104477614)(905,15.330930348258706)(1106,19.531940298507465)(1307,23.920925373134327)(1508,28.483)(1709,33.18437313432836)(1910,38.04372139303482)(2111,43.06094029850746)(2312,48.235)(2513,53.547522388059704)(2714,58.996686567164176)(2915,64.59732835820896)(3116,70.337)(3317,76.19528358208956)(3518,82.20574626865671)(3719,88.361)(3920,94.60893034825871)(4121,101.006)(4322,107.48502985074627)(4523,114.10801492537314)(4724,120.844)(4925,127.69661194029851)(5126,134.699)(5327,141.81149253731343)(5528,149.04737313432838)(5729,156.406)(5930,163.86664676616914)(6131,171.478)(6332,179.14918407960198)(6533,186.97)(6734,194.85467164179104)(6935,202.88731343283584)(7136,211.061)(7337,219.3411791044776)(7538,227.76917910447762)(7739,236.32)(7940,244.98292537313432)(8141,253.796)(8342,262.72932835820893)(8543,271.78371641791045)(8744,280.982)(8945,290.2904925373134)(9146,299.742104477612)(9347,309.341)(9548,319.0860298507463)(9749,328.98271641791047)(9950,339.03151243781093)(10151,349.232)(10352,359.57807462686566)(10553,370.0826616915423)(10754,380.74748258706467)(10955,391.57537810945274)(11156,402.5727462686568)(11357,413.753)(11558,425.1070895522388)(11759,436.67216417910447)(11960,448.458)(12161,460.41110447761196)(12362,472.56805970149253)(12563,484.93310447761195)(12764,497.5041641791045)(12965,510.3073631840796)(13166,523.3626119402985)(13367,536.6901094527364)(13568,550.285671641791)(13769,564.1436119402985)(13970,578.2974626865672)(14171,592.7570746268657)(14372,607.5117462686567)(14573,622.5840845771145)(14774,637.9978258706467)(14975,653.7967014925373)(15176,669.9916119402984)(15377,686.6112686567164)(15578,703.6454975124378)(15779,721.1384427860697)(15980,739.1880248756219)(16181,757.812527363184)(16382,777.0580995024876)(16583,797.0034129353235)(16784,817.7288557213931)(16985,839.3852139303483)(17186,861.9961393034826)(17387,885.653184079602)(17588,910.7359850746269)(17789,937.4878059701492)(17990,966.3835820895522)(18191,997.829039800995)(18392,1032.7199303482587)(18593,1072.2673681592041)(18794,1119.2120796019901)(18995,1179.8213482587066)(19196,1272.2905472636817)(19397,1467.4282437810944)(19588,2043.435502762431)(19680,2829.394)};
\addlegendentry{CVC4}

\addplot[name path=pathZ3Seq , colourZ3Seq, line width=1.5pt, densely dotted] coordinates {(1,0.018)(100,2.4031507537688443)(299,7.985452261306532)(498,14.293085427135678)(697,21.100899497487436)(896,28.286884422110553)(1095,35.819814070351754)(1294,43.68581407035175)(1493,51.834482412060304)(1692,60.27137688442211)(1891,68.95945226130654)(2090,77.90193969849247)(2289,87.10191959798995)(2488,96.54861809045227)(2687,106.22246231155779)(2886,116.13248743718593)(3085,126.26388442211055)(3284,136.6150351758794)(3483,147.16488442211053)(3682,157.91623618090452)(3881,168.87436683417084)(4080,180.0364824120603)(4279,191.3924974874372)(4478,202.93888442211053)(4677,214.6761256281407)(4876,226.61836683417084)(5075,238.78113567839196)(5274,251.14168844221103)(5473,263.6990904522613)(5672,276.45957286432156)(5871,289.42120603015076)(6070,302.5790904522613)(6269,315.9372060301508)(6468,329.4968442211055)(6667,343.27732663316584)(6866,357.2985075376884)(7065,371.5474623115578)(7264,386.03910050251255)(7463,400.76687939698496)(7662,415.72773869346736)(7861,430.93682412060303)(8060,446.3960351758794)(8259,462.13759798994977)(8458,478.14761306532665)(8657,494.4278844221106)(8856,510.99368341708544)(9055,527.8653768844222)(9254,545.0245025125628)(9453,562.4803668341709)(9652,580.2216984924623)(9851,598.2776783919597)(10050,616.661055276382)(10249,635.386688442211)(10448,654.4905829145729)(10647,673.9630703517587)(10846,693.7858894472362)(11045,713.996703517588)(11244,734.624)(11443,755.6792412060302)(11642,777.2059698492462)(11841,799.2096381909548)(12040,821.6998040201005)(12239,844.7337587939699)(12438,868.3513768844222)(12637,892.5619246231156)(12836,917.421809045226)(13035,942.9751859296482)(13234,969.2891306532663)(13433,996.4499396984924)(13632,1024.5357085427136)(13831,1053.6251658291458)(14030,1083.795231155779)(14229,1115.1595829145729)(14428,1147.8876733668342)(14627,1181.9985728643214)(14826,1217.6542060301506)(15025,1255.2216030150755)(15224,1294.8965728643216)(15423,1337.067743718593)(15622,1382.2271457286433)(15821,1431.0464974874371)(16020,1483.9062864321609)(16219,1541.339663316583)(16418,1604.2799296482413)(16617,1673.7163316582914)(16816,1750.595296482412)(17015,1836.1683417085428)(17214,1932.0301256281407)(17413,2041.5591005025126)(17612,2167.497527638191)(17811,2314.6906331658292)(18010,2490.1677889447237)(18209,2705.6372964824122)(18408,2979.368703517588)(18607,3330.8951557788946)(18806,3812.0925879396987)(19005,4503.1137085427135)(19204,5596.160979899498)(19398,7554.025391534392)(19494,9018.425)};
\addlegendentry{Z3Seq}

\addplot[name path=pathOSTRICH , colourOSTRICH, line width=1.5pt, dotted] coordinates {(1,0.975)(86.5,160.40451744186046)(258.5,568.371726744186)(430.5,1020.1490581395349)(602.5,1491.5828313953489)(774.5,1978.428203488372)(946.5,2478.8642790697672)(1118.5,2991.718319767442)(1290.5,3515.605668604651)(1462.5,4048.8564127906975)(1634.5,4591.111470930233)(1806.5,5142.068825581395)(1978.5,5700.826761627907)(2150.5,6266.638093023255)(2322.5,6839.02025)(2494.5,7417.995715116279)(2666.5,8003.350308139535)(2838.5,8595.512889534883)(3010.5,9194.500488372094)(3182.5,9799.89647093023)(3354.5,10411.233459302324)(3526.5,11028.528860465116)(3698.5,11651.785656976745)(3870.5,12280.580441860466)(4042.5,12915.285488372094)(4214.5,13555.690843023256)(4386.5,14201.43834883721)(4558.5,14852.78941860465)(4730.5,15509.399779069769)(4902.5,16171.249122093024)(5074.5,16838.424034883723)(5246.5,17510.412656976743)(5418.5,18187.344343023255)(5590.5,18869.548732558138)(5762.5,19556.864976744186)(5934.5,20249.381808139537)(6106.5,20946.958389534884)(6278.5,21649.58336627907)(6450.5,22357.77808139535)(6622.5,23071.195552325582)(6794.5,23789.736854651164)(6966.5,24513.63943604651)(7138.5,25242.65156395349)(7310.5,25977.21768604651)(7482.5,26717.431029069765)(7654.5,27463.165959302325)(7826.5,28214.375976744188)(7998.5,28971.820965116276)(8170.5,29736.11637790698)(8342.5,30506.708970930234)(8514.5,31283.620895348835)(8686.5,32067.447406976746)(8858.5,32858.46465697674)(9030.5,33657.41348255814)(9202.5,34463.56119186047)(9374.5,35277.24540697674)(9546.5,36098.36647093023)(9718.5,36927.83530813953)(9890.5,37766.38743023256)(10062.5,38615.028005813954)(10234.5,39473.9209883721)(10406.5,40343.89848837209)(10578.5,41224.68872093024)(10750.5,42116.23061627907)(10922.5,43018.07281395349)(11094.5,43932.07025581395)(11266.5,44859.086139534884)(11438.5,45800.00022674418)(11610.5,46757.24031976744)(11782.5,47731.63272093023)(11954.5,48725.404377906976)(12126.5,49739.11005232558)(12298.5,50773.38852325582)(12470.5,51829.505197674414)(12642.5,52910.343720930236)(12814.5,54015.646308139534)(12986.5,55146.71020930233)(13158.5,56306.85422093023)(13330.5,57499.62365697674)(13502.5,58727.06503488372)(13674.5,59989.85539534884)(13846.5,61290.397122093025)(14018.5,62629.64277906976)(14190.5,64012.68377906977)(14362.5,65443.07024418604)(14534.5,66930.27169186047)(14706.5,68484.83087790698)(14878.5,70109.76006976745)(15050.5,71812.81211627906)(15222.5,73599.28979651163)(15394.5,75471.45331395349)(15566.5,77427.86972093023)(15738.5,79466.93958139535)(15910.5,81611.69123255815)(16082.5,83888.74316860466)(16254.5,86323.82327906977)(16426.5,88945.34439534883)(16598.5,91778.41816860465)(16762.5,94732.11610256409)(16842,96248.051)};
\addlegendentry{OSTRICH}

\addplot[name path=pathZ3-Trau , colourZ3-Trau, line width=1.5pt,dash pattern={on 7pt off 2pt on 1pt off 3pt}] coordinates {(1,0.02)(55,1.5962660550458716)(164,5.296467889908256)(273,9.415275229357798)(382,13.858844036697247)(491,18.549073394495412)(600,23.447532110091743)(709,28.570752293577982)(818,33.89162385321101)(927,39.35345871559633)(1036,44.95164220183486)(1145,50.681532110091744)(1254,56.5404495412844)(1363,62.53289908256881)(1472,68.66492660550459)(1581,74.91783486238532)(1690,81.29174311926606)(1799,87.79111926605505)(1908,94.40571559633028)(2017,101.12526605504587)(2126,107.97308256880734)(2235,114.95152293577982)(2344,122.05584403669725)(2453,129.2725504587156)(2562,136.58008256880734)(2671,143.98291743119265)(2780,151.49549541284404)(2889,159.12216513761467)(2998,166.86928440366972)(3107,174.7324495412844)(3216,182.71653211009175)(3325,190.83343119266056)(3434,199.0868990825688)(3543,207.47312844036696)(3652,215.9771651376147)(3761,224.60355045871557)(3870,233.37045871559633)(3979,242.28308256880734)(4088,251.33661467889908)(4197,260.568128440367)(4306,269.9682752293578)(4415,279.51022018348624)(4524,289.20029357798165)(4633,299.0655229357798)(4742,309.1098440366972)(4851,319.327623853211)(4960,329.7172293577982)(5069,340.298871559633)(5178,351.0900917431193)(5287,362.0894495412844)(5396,373.31945871559634)(5505,384.7735412844036)(5614,396.4742018348624)(5723,408.4829082568807)(5832,420.81974311926604)(5941,433.46271559633027)(6050,446.44219266055046)(6159,459.8125871559633)(6268,473.59859633027526)(6377,487.8687247706422)(6486,502.6575137614679)(6595,517.9996880733945)(6704,533.9971100917431)(6813,550.8093486238532)(6922,568.4256697247706)(7031,586.8888532110092)(7140,606.3400366972478)(7249,626.8848073394495)(7358,648.4582844036697)(7467,671.3771100917431)(7576,695.9147247706422)(7685,722.3347247706422)(7794,751.1013394495413)(7903,782.6003302752295)(8012,817.3175779816514)(8121,855.3683944954129)(8230,897.7217889908256)(8339,945.1402844036697)(8448,998.4036513761469)(8557,1058.9789449541286)(8666,1127.7783027522937)(8775,1205.970376146789)(8884,1296.1327706422019)(8993,1399.8933211009173)(9102,1524.485724770642)(9211,1673.126844036697)(9320,1854.984853211009)(9429,2077.1858256880732)(9538,2350.382669724771)(9647,2693.2037981651374)(9756,3106.7357339449545)(9865,3598.6923302752293)(9974,4208.8423486238535)(10083,4945.459944954128)(10192,5780.918174311927)(10301,6738.831366972478)(10410,7911.462339449542)(10519,9369.216678899083)(10623,11060.527969696968)(10674,11995.347)};
\addlegendentry{Z3-Trau}

\addplot[name path=pathZ3str3 , colourZ3str3, line width=1.5pt,dash dot] coordinates {(1,0.013)(69.5,1.3421521739130435)(207.5,4.231333333333333)(345.5,7.407420289855072)(483.5,10.861166666666666)(621.5,14.597144927536233)(759.5,18.629594202898552)(897.5,22.961195652173913)(1035.5,27.587072463768116)(1173.5,32.47027536231884)(1311.5,37.64997826086957)(1449.5,43.067282608695656)(1587.5,48.72817391304348)(1725.5,54.58302173913043)(1863.5,60.6218768115942)(2001.5,66.8345)(2139.5,73.22502173913043)(2277.5,79.76656521739132)(2415.5,86.458884057971)(2553.5,93.30839855072465)(2691.5,100.28967391304347)(2829.5,107.42465217391305)(2967.5,114.711)(3105.5,122.1363115942029)(3243.5,129.6967391304348)(3381.5,137.42286956521738)(3519.5,145.29404347826087)(3657.5,153.30132608695652)(3795.5,161.45952173913042)(3933.5,169.7475072463768)(4071.5,178.1752391304348)(4209.5,186.7503768115942)(4347.5,195.46583333333334)(4485.5,204.34184782608693)(4623.5,213.38526086956523)(4761.5,222.57160869565217)(4899.5,231.90326086956523)(5037.5,241.3661739130435)(5175.5,250.97883333333334)(5313.5,260.7344420289855)(5451.5,270.6438333333333)(5589.5,280.7121594202899)(5727.5,290.9488695652174)(5865.5,301.3385434782609)(6003.5,311.8785217391304)(6141.5,322.5833115942029)(6279.5,333.4611086956522)(6417.5,344.5050652173913)(6555.5,355.71197826086956)(6693.5,367.10767391304347)(6831.5,378.6940652173913)(6969.5,390.475615942029)(7107.5,402.44832608695657)(7245.5,414.61139855072463)(7383.5,427.0071376811594)(7521.5,439.64460869565215)(7659.5,452.50354347826084)(7797.5,465.608963768116)(7935.5,478.95010869565215)(8073.5,492.5117971014493)(8211.5,506.31465217391303)(8349.5,520.3938985507247)(8487.5,534.7307246376812)(8625.5,549.3698478260869)(8763.5,564.2815072463768)(8901.5,579.4927463768116)(9039.5,595.0143985507246)(9177.5,610.860115942029)(9315.5,627.0620942028986)(9453.5,643.603768115942)(9591.5,660.537420289855)(9729.5,677.8826811594203)(9867.5,695.6639855072464)(10005.5,713.9109347826086)(10143.5,732.6206956521739)(10281.5,751.8423115942029)(10419.5,771.6738188405797)(10557.5,792.1500724637681)(10695.5,813.2650942028986)(10833.5,835.1250942028986)(10971.5,857.8826304347826)(11109.5,881.5716449275362)(11247.5,906.2355942028986)(11385.5,932.1261811594203)(11523.5,959.4113623188406)(11661.5,988.3594420289854)(11799.5,1019.2943478260869)(11937.5,1052.8078043478263)(12075.5,1089.4063623188406)(12213.5,1130.4668043478262)(12351.5,1178.123384057971)(12489.5,1235.1250144927537)(12627.5,1312.1216594202897)(12765.5,1454.7297971014493)(12903.5,1748.689804347826)(13041.5,2145.823347826087)(13179.5,2607.386695652174)(13317.5,3265.4395289855074)(13455.5,4661.4111014492755)(13529.5,5858.6839)(13536,5968.054)};
\addlegendentry{Z3str3}

\addplot[name path=pathZ3str3RE , colourZ3str3RE, line width=1.5pt] coordinates {(1,0.007)(102,0.8843793103448275)(305,2.900591133004926)(508,5.1892807881773395)(711,7.640428571428571)(914,10.246)(1117,12.90135960591133)(1320,15.707)(1523,18.549)(1726,21.401891625615765)(1929,24.401)(2132,27.446)(2335,30.491)(2538,33.60024137931035)(2741,36.844)(2944,40.092)(3147,43.36793596059113)(3350,46.796)(3553,50.247)(3756,53.698)(3959,57.18717733990148)(4162,60.826)(4365,64.48)(4568,68.134)(4771,71.788)(4974,75.544)(5177,79.401)(5380,83.258)(5583,87.115)(5786,90.97210344827586)(5989,94.937)(6192,98.997)(6395,103.057)(6598,107.117)(6801,111.22631527093596)(7004,115.48)(7207,119.743)(7410,124.006)(7613,128.32042857142858)(7816,132.778)(8019,137.244)(8222,141.72929064039408)(8425,146.366)(8628,151.035)(8831,155.74156650246303)(9034,160.598)(9237,165.47066995073894)(9440,170.46)(9643,175.535)(9846,180.68586206896552)(10049,185.96210344827585)(10252,191.348)(10455,196.83924630541873)(10658,202.476)(10861,208.19635960591134)(11064,214.07209852216747)(11267,220.101)(11470,226.24101477832514)(11673,232.53204926108373)(11876,238.97304926108373)(12079,245.575)(12282,252.32614285714288)(12485,259.2397290640394)(12688,266.315)(12891,273.548)(13094,280.9327290640394)(13297,288.4885665024631)(13500,296.2226206896552)(13703,304.1382906403941)(13906,312.2320344827586)(14109,320.50103940886703)(14312,328.965842364532)(14515,337.6365862068966)(14718,346.5164778325123)(14921,355.576842364532)(15124,364.8534482758621)(15327,374.34124630541874)(15530,384.0659359605911)(15733,394.04527093596056)(15936,404.3071428571428)(16139,414.8373842364532)(16342,425.65860591133)(16545,436.7794137931034)(16748,448.2154236453202)(16951,460.03211330049265)(17154,472.2825024630542)(17357,484.9897339901478)(17560,498.22120197044336)(17763,512.0137684729063)(17966,526.4703103448276)(18169,541.744748768473)(18372,557.971472906404)(18575,575.284487684729)(18778,594.1859458128079)(18981,615.3810640394088)(19184,640.7673349753695)(19387,675.6836699507389)(19590,746.8065714285715)(19774.5,1022.074578313253)(19859,1525.15)};
\addlegendentry{Z3str3RE}

\path[name path=axisCVC4] (axis cs:0,0) -- (axis cs:19680,0);
\addplot [thick,color=colourCVC4,fill=colourCVC4,fill opacity=0.1] fill between [of=pathCVC4 and axisCVC4];
\path[name path=axisZ3Seq] (axis cs:0,0) -- (axis cs:19494,0);
\addplot [thick,color=colourZ3Seq,fill=colourZ3Seq,fill opacity=0.1] fill between [of=pathZ3Seq and axisZ3Seq];
\path[name path=axisOSTRICH] (axis cs:0,0) -- (axis cs:16842,0);
\addplot [thick,color=colourOSTRICH,fill=colourOSTRICH,fill opacity=0.1] fill between [of=pathOSTRICH and axisOSTRICH];
\path[name path=axisZ3-Trau] (axis cs:0,0) -- (axis cs:10674,0);
\addplot [thick,color=colourZ3-Trau,fill=colourZ3-Trau,fill opacity=0.1] fill between [of=pathZ3-Trau and axisZ3-Trau];
\path[name path=axisZ3str3] (axis cs:0,0) -- (axis cs:13536,0);
\addplot [thick,color=colourZ3str3,fill=colourZ3str3,fill opacity=0.1] fill between [of=pathZ3str3 and axisZ3str3];
\path[name path=axisZ3str3RE] (axis cs:0,0) -- (axis cs:19859,0);
\addplot [thick,color=colourZ3str3RE,fill=colourZ3str3RE,fill opacity=0.1] fill between [of=pathZ3str3RE and axisZ3str3RE];
\end{axis}\end{tikzpicture}}

\caption{Cactus plot summarizing detailed performance on Automatark benchmark.}
\label{fig:cactus_automatark}

\end{figure}

\begin{table}[t!]

\resizebox{0.95\textwidth}{!}{
\begin{tabular}{|c |c |c |c |c |c |c |}
\hline
&CVC4&Z3Seq&OSTRICH&Z3-Trau&Z3str3&Z3str3RE\\ 
  \hline\hline 
sat &14376&14204&11461&8157&9151&\textbf{14437}\\ 
 \hline
unsat &5304&5290&5381&3817&4385&\textbf{5422}\\ 
 \hline
\hline 
 unknown &1&\textbf{0}&15&5045&406&\textbf{0}\\ 
 \hline
timeout &298&485&3122&2960&6037&\textbf{120}\\ 
 \hline
soundness error &\textbf{0}&\textbf{0}&\textbf{0}&1300&\textbf{0}&\textbf{0}\\ 
 \hline
program crashes &\textbf{0}&\textbf{0}&\textbf{0}&1063&2&\textbf{0}\\ 
 \hline
\hline 
 Total correct &19680&19494&16842&10674&13536&\textbf{19859}\\ 
 \hline
 Contribution score &1.0&1.0&\textbf{2.0}&--&0.0&0.5\\ 
 \hline
Time (s) &8789.425&18718.425&158910.126&80021.352&126825.967&\textbf{3925.150}\\ 
 \hline
Time w/o timeouts (s) &2829.425&9018.425&96470.126&20821.352&6085.967&\textbf{1525.150}\\ 
 \hline
\end{tabular}}
\vspace{0.1cm}
\caption{Detailed results for the Automatark benchmark. \toolname{} has the biggest lead with a score of 1.01.}
\label{tab:cactus_automatark}

\end{table}

%% file: data2.tex
\begin{figure}[t!]

\resizebox{.95\textwidth}{!}{\pgfplotsset{scaled x ticks=false}\pgfplotsset{scaled y ticks=false}\begin{tikzpicture}\begin{axis}[title=Stringfuzz RegEx Generated,xmin=-250,xlabel=Solved instances,ylabel=Time (seconds),,legend columns=2,legend style={nodes={scale=1, transform shape}, fill=none,anchor=east,align=center },axis line style={draw=none}, xtick pos=left, ytick pos=left, ymajorgrids=true, legend style={at={(1,0.9)},draw=none,fill=white},x post scale=2,y post scale=1.25,ymax=6500]
\addplot[name path=pathCVC4 , colourCVC4, line width=1.5pt,dashed] coordinates {(1,0.008)(14.5,0.18510714285714286)(42.5,0.6708571428571429)(70.5,1.29125)(98.5,2.0128214285714288)(126.5,2.8293928571428575)(154.5,3.7286785714285715)(182.5,4.71975)(210.5,5.7905)(238.5,6.923357142857143)(266.5,8.131357142857143)(294.5,9.419142857142857)(322.5,10.78132142857143)(350.5,12.203035714285713)(378.5,13.687357142857143)(406.5,15.23967857142857)(434.5,16.86225)(462.5,18.55314285714286)(490.5,20.31039285714286)(518.5,22.1445)(546.5,24.05125)(574.5,26.00782142857143)(602.5,28.031964285714285)(630.5,30.13875)(658.5,32.31482142857143)(686.5,34.54828571428572)(714.5,36.84782142857143)(742.5,39.23392857142857)(770.5,41.701)(798.5,44.247928571428574)(826.5,46.868)(854.5,49.55985714285715)(882.5,52.32192857142857)(910.5,55.16907142857143)(938.5,58.099642857142854)(966.5,61.10275)(994.5,64.17560714285715)(1022.5,67.32753571428572)(1050.5,70.57860714285715)(1078.5,73.90725)(1106.5,77.30139285714286)(1134.5,80.77125)(1162.5,84.32507142857143)(1190.5,87.963)(1218.5,91.68614285714285)(1246.5,95.51614285714285)(1274.5,99.46303571428571)(1302.5,103.52810714285714)(1330.5,107.70878571428571)(1358.5,111.99457142857143)(1386.5,116.41339285714285)(1414.5,120.96367857142856)(1442.5,125.66007142857144)(1470.5,130.50717857142857)(1498.5,135.5244642857143)(1526.5,140.71)(1554.5,146.0490357142857)(1582.5,151.56278571428572)(1610.5,157.242)(1638.5,163.0927142857143)(1666.5,169.11564285714286)(1694.5,175.29989285714288)(1722.5,181.66310714285714)(1750.5,188.27167857142857)(1778.5,195.1444642857143)(1806.5,202.3135714285714)(1834.5,209.76896428571428)(1862.5,217.46946428571428)(1890.5,225.4594285714286)(1918.5,233.82425)(1946.5,242.49282142857143)(1974.5,251.58585714285712)(2002.5,261.30925)(2030.5,271.56592857142857)(2058.5,282.3023928571428)(2086.5,293.60914285714284)(2114.5,305.59353571428574)(2142.5,318.38275)(2170.5,332.28975)(2198.5,347.4446785714286)(2226.5,363.69135714285716)(2254.5,380.97739285714283)(2282.5,399.40475)(2310.5,419.0940357142857)(2338.5,440.28085714285714)(2366.5,463.9591785714286)(2394.5,490.226)(2422.5,519.156)(2450.5,552.3558214285715)(2478.5,590.3926428571428)(2506.5,633.5361428571429)(2534.5,692.11425)(2562.5,781.3376428571429)(2590.5,912.2188214285715)(2618.5,1100.2021785714287)(2646.5,1358.4184642857142)(2674.5,1680.0195357142857)(2702.5,2062.139107142857)(2730.5,2507.012714285714)(2750.5,2869.938)(2758,2996.207)};
\addlegendentry{CVC4}

\addplot[name path=pathZ3Seq , colourZ3Seq, line width=1.5pt, densely dotted] coordinates {(1,0.02)(14.5,0.41446428571428573)(42.5,1.370857142857143)(70.5,2.4968214285714283)(98.5,3.7115)(126.5,5.015357142857143)(154.5,6.383964285714286)(182.5,7.815785714285715)(210.5,9.329107142857143)(238.5,10.909928571428571)(266.5,12.552535714285714)(294.5,14.25757142857143)(322.5,16.021285714285714)(350.5,17.8445)(378.5,19.72882142857143)(406.5,21.672964285714286)(434.5,23.668178571428573)(462.5,25.72189285714286)(490.5,27.84389285714286)(518.5,30.051071428571426)(546.5,32.33460714285714)(574.5,34.70128571428572)(602.5,37.14342857142857)(630.5,39.666428571428575)(658.5,42.269535714285716)(686.5,44.952642857142855)(714.5,47.72871428571428)(742.5,50.601642857142856)(770.5,53.56557142857143)(798.5,56.61125)(826.5,59.78639285714286)(854.5,63.05157142857143)(882.5,66.40435714285715)(910.5,69.8725)(938.5,73.45010714285715)(966.5,77.14639285714286)(994.5,80.95482142857144)(1022.5,84.8812142857143)(1050.5,88.94817857142857)(1078.5,93.1655)(1106.5,97.53060714285715)(1134.5,102.02985714285714)(1162.5,106.69532142857143)(1190.5,111.53189285714285)(1218.5,116.56175)(1246.5,121.76335714285715)(1274.5,127.14082142857143)(1302.5,132.6914642857143)(1330.5,138.41935714285714)(1358.5,144.38310714285714)(1386.5,150.65835714285714)(1414.5,157.27889285714286)(1442.5,164.25514285714286)(1470.5,171.67260714285712)(1498.5,179.57292857142858)(1526.5,188.02189285714286)(1554.5,197.112)(1582.5,206.9492857142857)(1610.5,217.56225)(1638.5,229.28389285714286)(1666.5,242.23957142857142)(1694.5,256.2975)(1722.5,271.7232142857143)(1750.5,288.6799642857143)(1778.5,307.4274285714286)(1806.5,327.983)(1834.5,350.83567857142856)(1862.5,376.1066071428572)(1890.5,404.39810714285716)(1918.5,436.3674285714286)(1946.5,472.32460714285713)(1974.5,512.0207857142857)(2002.5,556.1057142857143)(2030.5,605.1380714285715)(2058.5,660.0047857142857)(2086.5,721.0678571428572)(2114.5,788.6384642857143)(2142.5,865.6264642857143)(2170.5,954.14075)(2198.5,1052.4723214285714)(2226.5,1159.876642857143)(2254.5,1276.737392857143)(2282.5,1407.0914285714287)(2310.5,1550.0000714285713)(2338.5,1705.819107142857)(2366.5,1877.6885)(2394.5,2072.020571428571)(2422.5,2290.8108571428575)(2450.5,2534.534714285714)(2478.5,2802.9885357142857)(2506.5,3096.805428571429)(2534.5,3422.2965)(2562.5,3777.5361071428574)(2590.5,4159.514714285714)(2618.5,4576.67)(2646.5,5032.576857142857)(2674.5,5536.706321428572)(2692.5,5879.912875)(2698,5969.0)};
\addlegendentry{Z3Seq}

\addplot[name path=pathOSTRICH , colourOSTRICH, line width=1.5pt, dotted] coordinates {(1,0.916)(15,25.611758620689656)(44,89.13810344827586)(73,169.6659655172414)(102,257.64806896551727)(131,351.0540344827586)(160,450.0902068965517)(189,554.1508965517241)(218,663.7524827586207)(247,780.222)(276,902.884551724138)(305,1031.2963448275862)(334,1166.6751034482759)(363,1307.6414827586207)(392,1453.3730344827586)(421,1603.7649310344827)(450,1758.1052068965516)(479,1915.8070689655171)(508,2077.1165172413794)(537,2241.7791379310343)(566,2409.5376896551725)(595,2580.474482758621)(624,2753.826620689655)(653,2929.3597931034483)(682,3107.483655172414)(711,3287.855)(740,3470.3854137931035)(769,3654.9573103448274)(798,3841.743275862069)(827,4030.514482758621)(856,4221.702620689654)(885,4415.035793103448)(914,4610.404275862069)(943,4807.972)(972,5007.903931034483)(1001,5210.5672758620685)(1030,5415.673586206896)(1059,5622.770931034483)(1088,5831.630482758621)(1117,6042.692551724138)(1146,6256.317310344827)(1175,6472.665620689655)(1204,6690.955551724138)(1233,6911.284827586207)(1262,7133.647827586206)(1291,7357.902172413794)(1320,7584.247655172414)(1349,7812.630448275862)(1378,8043.305206896552)(1407,8275.847275862068)(1436,8510.69220689655)(1465,8747.747379310345)(1494,8987.266793103448)(1523,9229.187620689654)(1552,9474.28572413793)(1581,9722.91548275862)(1610,9974.162655172413)(1639,10227.644034482759)(1668,10483.188310344829)(1697,10741.189137931035)(1726,11001.320310344829)(1755,11264.116103448276)(1784,11529.910482758622)(1813,11798.687103448276)(1842,12070.358965517242)(1871,12345.242620689654)(1900,12623.162896551725)(1929,12904.812482758622)(1958,13190.415241379309)(1987,13479.517310344829)(2016,13771.906344827587)(2045,14068.32324137931)(2074,14369.143862068966)(2103,14674.322896551725)(2132,14984.12427586207)(2161,15300.156172413792)(2190,15622.289896551725)(2219,15950.549931034482)(2248,16283.756689655173)(2277,16621.76255172414)(2306,16966.473103448276)(2335,17317.183172413792)(2364,17674.096827586207)(2393,18039.55203448276)(2422,18411.8874137931)(2451,18791.26003448276)(2480,19177.92803448276)(2509,19572.955724137933)(2538,19975.895206896552)(2567,20389.73824137931)(2596,20812.907620689657)(2625,21247.049965517243)(2654,21692.771896551723)(2683,22152.172655172417)(2712,22626.567793103448)(2741,23121.08255172414)(2770,23638.54655172414)(2799,24181.74175862069)(2818,24552.308777777776)(2824,24651.865)};
\addlegendentry{OSTRICH}

\addplot[name path=pathZ3-Trau , colourZ3-Trau, line width=1.5pt,dash pattern={on 7pt off 2pt on 1pt off 3pt}] coordinates {(1,0.019)(13,0.38027999999999995)(38,1.29484)(63,2.42508)(88,3.6981599999999997)(113,5.101520000000001)(138,6.6317200000000005)(163,8.28896)(188,10.056280000000001)(213,12.00056)(238,14.168280000000001)(263,16.539080000000002)(288,19.14852)(313,22.01712)(338,25.22324)(363,28.89216)(388,33.037839999999996)(413,37.832080000000005)(438,43.28252)(463,49.46604)(488,56.34988)(513,63.948440000000005)(538,72.62344)(563,82.3686)(588,93.08664)(613,104.51248)(638,116.67775999999999)(663,129.82728)(688,143.9928)(713,158.91388)(738,174.61108)(763,191.15879999999999)(788,208.64644)(813,227.09784)(838,246.39632)(863,266.736)(888,288.12944)(913,310.73684000000003)(938,334.60752)(963,359.66564)(988,385.75696000000005)(1013,412.94867999999997)(1038,441.28592)(1063,470.96296)(1088,501.71528)(1113,533.57084)(1138,566.68188)(1163,601.1162800000001)(1188,636.71952)(1213,673.41412)(1238,711.09428)(1263,750.11412)(1288,790.464)(1313,832.43648)(1338,875.81268)(1363,920.6149200000001)(1388,966.70844)(1413,1013.8771999999999)(1438,1062.27944)(1463,1111.9296399999998)(1488,1163.13944)(1513,1215.90508)(1538,1270.0818000000002)(1563,1326.20164)(1588,1384.2681599999999)(1613,1445.29676)(1638,1508.97652)(1663,1574.4766000000002)(1688,1641.6694)(1713,1710.62016)(1738,1781.5320800000002)(1763,1854.70172)(1788,1929.79824)(1813,2007.05508)(1838,2086.74468)(1863,2168.78232)(1888,2253.49712)(1913,2341.91128)(1938,2433.7496)(1963,2528.7987200000002)(1988,2627.98412)(2013,2733.10776)(2038,2843.77536)(2063,2959.22208)(2088,3080.2014)(2113,3207.86984)(2138,3341.3264)(2163,3481.3294)(2188,3630.66228)(2213,3790.5126800000003)(2238,3967.24428)(2263,4166.1294)(2288,4388.42056)(2313,4637.8542)(2338,4917.10948)(2363,5249.878839999999)(2388,5658.87276)(2402.5,5926.561)(2406,5976.028)};
\addlegendentry{Z3-Trau}

\addplot[name path=pathZ3str3 , colourZ3str3, line width=1.5pt,dash dot] coordinates {(1,0.017)(17,0.37324242424242426)(50,1.1908181818181818)(83,2.1055757575757577)(116,3.140848484848485)(149,4.308606060606061)(182,5.631)(215,7.109606060606061)(248,8.742515151515152)(281,10.477818181818183)(314,12.294272727272729)(347,14.197030303030305)(380,16.185727272727274)(413,18.254272727272728)(446,20.402757575757576)(479,22.636454545454544)(512,24.95087878787879)(545,27.352454545454545)(578,29.832454545454546)(611,32.38781818181818)(644,35.03269696969697)(677,37.7620303030303)(710,40.57606060606061)(743,43.48139393939394)(776,46.48609090909091)(809,49.585454545454546)(842,52.77393939393939)(875,56.075818181818185)(908,59.49103030303031)(941,62.99612121212122)(974,66.6010303030303)(1007,70.32257575757576)(1040,74.18612121212122)(1073,78.21542424242423)(1106,82.40718181818183)(1139,86.73339393939393)(1172,91.20957575757576)(1205,95.83557575757575)(1238,100.6280303030303)(1271,105.61384848484848)(1304,110.82915151515152)(1337,116.25209090909091)(1370,121.87751515151515)(1403,127.78139393939394)(1436,134.01324242424243)(1469,140.53933333333333)(1502,147.47378787878787)(1535,154.78875757575756)(1568,162.64690909090908)(1601,171.15233333333333)(1634,180.30066666666664)(1667,190.17184848484848)(1700,200.75257575757576)(1733,212.11730303030305)(1766,224.24927272727274)(1799,237.42939393939395)(1832,251.89972727272726)(1865,267.7590606060606)(1898,285.157696969697)(1931,304.27660606060607)(1964,325.2841212121212)(1997,348.1041818181818)(2030,372.63203030303026)(2063,399.07972727272727)(2096,427.39090909090913)(2129,457.8640303030303)(2162,490.95057575757573)(2195,526.5923333333334)(2228,564.6906060606061)(2261,605.0688181818182)(2294,647.4693333333333)(2327,692.6966363636363)(2360,740.7006666666666)(2393,791.0806363636364)(2426,843.6021515151515)(2459,898.2147878787879)(2492,955.3053636363636)(2525,1014.7820909090909)(2558,1076.2961515151517)(2591,1140.2373333333333)(2624,1206.8754848484848)(2657,1275.9172424242424)(2690,1348.4895757575757)(2723,1424.394878787879)(2756,1503.1083333333333)(2789,1584.937)(2822,1670.1566666666668)(2855,1759.7423636363635)(2888,1855.502606060606)(2921,1957.4345151515151)(2954,2065.1422727272725)(2987,2180.4795757575757)(3020,2303.903484848485)(3053,2437.9494848484846)(3086,2588.140121212121)(3119,2753.9616666666666)(3152,2937.982575757576)(3185,3145.8734545454545)(3218,3396.6486060606057)(3246,3662.820086956522)(3259,3811.636)};
\addlegendentry{Z3str3}

\addplot[name path=pathZ3str3RE , colourZ3str3RE, line width=1.5pt] coordinates {(1,0.009)(21,0.2468780487804878)(62,0.7986341463414635)(103,1.4260975609756097)(144,2.104)(185,2.8200243902439026)(226,3.5693414634146343)(267,4.363609756097562)(308,5.2091463414634145)(349,6.094073170731707)(390,7.019365853658536)(431,7.990219512195122)(472,9.009731707317073)(513,10.079609756097561)(554,11.203756097560976)(595,12.382731707317072)(636,13.624682926829268)(677,14.947146341463414)(718,16.322341463414634)(759,17.75331707317073)(800,19.25031707317073)(841,20.820536585365854)(882,22.47580487804878)(923,24.208926829268293)(964,26.020975609756096)(1005,27.901121951219512)(1046,29.845219512195122)(1087,31.853560975609756)(1128,33.93146341463415)(1169,36.09334146341463)(1210,38.33126829268293)(1251,40.65360975609756)(1292,43.06868292682927)(1333,45.57290243902439)(1374,48.161195121951216)(1415,50.82124390243903)(1456,53.556634146341466)(1497,56.35573170731707)(1538,59.226975609756096)(1579,62.16873170731707)(1620,65.18770731707316)(1661,68.29358536585366)(1702,71.48717073170732)(1743,74.77575609756099)(1784,78.16526829268292)(1825,81.67663414634146)(1866,85.2999024390244)(1907,89.02651219512195)(1948,92.86243902439024)(1989,96.80456097560976)(2030,100.8620243902439)(2071,105.03187804878048)(2112,109.31990243902439)(2153,113.74356097560975)(2194,118.32643902439024)(2235,123.08478048780488)(2276,128.026756097561)(2317,133.21870731707315)(2358,138.6918536585366)(2399,144.37958536585364)(2440,150.34951219512197)(2481,156.7890731707317)(2522,163.82229268292684)(2563,171.7087317073171)(2604,180.66031707317075)(2645,190.99363414634146)(2686,203.14134146341462)(2727,217.288756097561)(2768,233.40143902439024)(2809,251.36670731707315)(2850,271.41053658536583)(2891,293.518756097561)(2932,317.642)(2973,344.1635365853658)(3014,372.9222195121951)(3055,404.1549756097561)(3096,438.50307317073174)(3137,476.1787804878049)(3178,516.8412682926829)(3219,560.6712439024391)(3260,607.1970731707316)(3301,656.7179756097561)(3342,709.9001707317074)(3383,766.982243902439)(3424,827.7085853658537)(3465,891.8164390243902)(3506,961.8189756097561)(3547,1038.0100243902439)(3588,1119.8172195121952)(3629,1207.8656097560977)(3670,1302.087243902439)(3711,1402.3393414634147)(3752,1509.160756097561)(3793,1623.7879268292681)(3834,1750.5318536585366)(3875,1892.1507317073172)(3916,2050.6085609756096)(3957,2232.1268780487803)(3998,2445.539926829268)(4039,2718.879975609756)(4061,2936.456)};
\addlegendentry{Z3str3RE}

\path[name path=axisCVC4] (axis cs:0,0) -- (axis cs:2758,0);
\addplot [thick,color=colourCVC4,fill=colourCVC4,fill opacity=0.1] fill between [of=pathCVC4 and axisCVC4];
\path[name path=axisZ3Seq] (axis cs:0,0) -- (axis cs:2698,0);
\addplot [thick,color=colourZ3Seq,fill=colourZ3Seq,fill opacity=0.1] fill between [of=pathZ3Seq and axisZ3Seq];
\path[name path=axisOSTRICH] (axis cs:0,0) -- (axis cs:2824,0);
\addplot [thick,color=colourOSTRICH,fill=colourOSTRICH,fill opacity=0.1] fill between [of=pathOSTRICH and axisOSTRICH];
\path[name path=axisZ3-Trau] (axis cs:0,0) -- (axis cs:2406,0);
\addplot [thick,color=colourZ3-Trau,fill=colourZ3-Trau,fill opacity=0.1] fill between [of=pathZ3-Trau and axisZ3-Trau];
\path[name path=axisZ3str3] (axis cs:0,0) -- (axis cs:3259,0);
\addplot [thick,color=colourZ3str3,fill=colourZ3str3,fill opacity=0.1] fill between [of=pathZ3str3 and axisZ3str3];
\path[name path=axisZ3str3RE] (axis cs:0,0) -- (axis cs:4061,0);
\addplot [thick,color=colourZ3str3RE,fill=colourZ3str3RE,fill opacity=0.1] fill between [of=pathZ3str3RE and axisZ3str3RE];
\end{axis}\end{tikzpicture}}

\caption{Cactus plot showing detailed results for the \benchmarkgen{} benchmark.}
\label{fig:cactus_generated}

\end{figure}
\begin{table}[t!]

\resizebox{0.95\textwidth}{!}{
\begin{tabular}{|c |c |c |c |c |c |c |}
\hline
&CVC4&Z3Seq&OSTRICH&Z3-Trau&Z3str3&Z3str3RE\\ 
  \hline\hline 
sat &2316&2001&2005&1590&3227&\textbf{3231}\\ 
 \hline
unsat &442&697&819&824&32&\textbf{830}\\ 
 \hline
\hline 
 unknown &\textbf{0}&\textbf{0}&1&192&\textbf{0}&\textbf{0}\\ 
 \hline
timeout &1412&1472&1345&1564&911&\textbf{109}\\ 
 \hline
soundness error &\textbf{0}&\textbf{0}&\textbf{0}&8&\textbf{0}&\textbf{0}\\ 
 \hline
program crashes &\textbf{0}&\textbf{0}&\textbf{0}&192&\textbf{0}&\textbf{0}\\ 
 \hline
\hline 
 Total correct &2758&2698&2824&2406&3259&\textbf{4061}\\ 
 \hline
Contribution score &0.0&\textbf{3.17}&2.0&--&0.0&0.17\\ 
\hline
Time (s) &31236.207&35409.000&51571.800&37323.550&22031.636&\textbf{5116.456}\\ 
 \hline
Time w/o timeouts (s) &2996.207&5969.000&24671.800&6043.550&3811.636&\textbf{2936.456}\\ 
 \hline
\end{tabular}}
\vspace{0.1cm}
\caption{Detailed results for the \benchmarkgen{} benchmark. \toolname{} has the biggest lead with a score of 1.25.}
\label{tab:cactus_generated}

\end{table}

%% file: data3.tex
\begin{figure}[t!]

\resizebox{.95\textwidth}{!}{\pgfplotsset{scaled x ticks=false}\pgfplotsset{scaled y ticks=false}\begin{tikzpicture}\begin{axis}[title=Stringfuzz RegEx Transformed,xmin=-250,xlabel=Solved instances,ylabel=Time (seconds),,legend columns=2,legend style={nodes={scale=1, transform shape}, fill=none,anchor=east,align=center },axis line style={draw=none}, xtick pos=left, ytick pos=left, ymajorgrids=true, legend style={draw=none,at={(1.05,0.95)},fill=white},x post scale=2,y post scale=1.25,ymax=1000]
\addplot[name path=pathCVC4 , colourCVC4, line width=1.5pt,dashed] coordinates {(1,0.007)(54.5,0.5303611111111111)(162.5,1.6835)(270.5,2.9210925925925926)(378.5,4.217)(486.5,5.538694444444444)(594.5,6.9375)(702.5,8.3415)(810.5,9.783416666666666)(918.5,11.294)(1026.5,12.806)(1134.5,14.3715)(1242.5,15.9915)(1350.5,17.6115)(1458.5,19.246805555555554)(1566.5,20.963)(1674.5,22.691)(1782.5,24.421555555555553)(1890.5,26.2245)(1998.5,28.0605)(2106.5,29.8965)(2214.5,31.77384259259259)(2322.5,33.717)(2430.5,35.661)(2538.5,37.61375925925926)(2646.5,39.6465)(2754.5,41.6985)(2862.5,43.756009259259265)(2970.5,45.891)(3078.5,48.051)(3186.5,50.211)(3294.5,52.41411111111111)(3402.5,54.6815)(3510.5,56.9495)(3618.5,59.259722222222216)(3726.5,61.635)(3834.5,64.011)(3942.5,66.41200925925925)(4050.5,68.8905)(4158.5,71.3745)(4266.5,73.87600925925925)(4374.5,76.458)(4482.5,79.05)(4590.5,81.66500925925925)(4698.5,84.3585)(4806.5,87.05922222222222)(4914.5,89.825)(5022.5,92.633)(5130.5,95.44716666666667)(5238.5,98.3395)(5346.5,101.2555)(5454.5,104.19583333333333)(5562.5,107.214)(5670.5,110.2493425925926)(5778.5,113.3655)(5886.5,116.4975925925926)(5994.5,119.688)(6102.5,122.928)(6210.5,126.17077777777779)(6318.5,129.4865)(6426.5,132.83491666666666)(6534.5,136.246)(6642.5,139.70209259259258)(6750.5,143.2165)(6858.5,146.7805)(6966.5,150.39505555555556)(7074.5,154.07350925925928)(7182.5,157.8305)(7290.5,161.6303611111111)(7398.5,165.51)(7506.5,169.44663888888888)(7614.5,173.44833333333335)(7722.5,177.52861111111113)(7830.5,181.6975)(7938.5,185.93875925925926)(8046.5,190.27547222222222)(8154.5,194.70638888888888)(8262.5,199.2295)(8370.5,203.84544444444444)(8478.5,208.56261111111112)(8586.5,213.3804722222222)(8694.5,218.31675925925927)(8802.5,223.40130555555555)(8910.5,228.6172222222222)(9018.5,233.99158333333335)(9126.5,239.56277777777777)(9234.5,245.3262777777778)(9342.5,251.2915648148148)(9450.5,257.4951111111111)(9558.5,264.0128055555556)(9666.5,270.9046666666667)(9774.5,278.22501851851854)(9882.5,286.0976296296296)(9990.5,294.6800185185185)(10098.5,304.4574259259259)(10206.5,316.24451851851853)(10314.5,331.00245370370374)(10422.5,352.82024074074076)(10516,389.69560759493675)(10557,469.643)};
\addlegendentry{CVC4}

\addplot[name path=pathZ3Seq , colourZ3Seq, line width=1.5pt, densely dotted] coordinates {(1,0.017)(54.5,1.123398148148148)(162.5,3.464888888888889)(270.5,5.9225)(378.5,8.433592592592593)(486.5,11.021)(594.5,13.634722222222223)(702.5,16.3275)(810.5,19.04125)(918.5,21.836)(1026.5,24.644)(1134.5,27.50063888888889)(1242.5,30.4165)(1350.5,33.34009259259259)(1458.5,36.343)(1566.5,39.367)(1674.5,42.409666666666666)(1782.5,45.5325)(1890.5,48.6645)(1998.5,51.81575925925926)(2106.5,55.047)(2214.5,58.287)(2322.5,61.54508333333334)(2430.5,64.8835)(2538.5,68.2315)(2646.5,71.6208425925926)(2754.5,75.076)(2862.5,78.532)(2970.5,82.03200925925925)(3078.5,85.5955)(3186.5,89.1595)(3294.5,92.76841666666667)(3402.5,96.44)(3510.5,100.112)(3618.5,103.8038611111111)(3726.5,107.5755)(3834.5,111.3555)(3942.5,115.15597222222222)(4050.5,119.036)(4158.5,122.92400925925925)(4266.5,126.8675)(4374.5,130.8635)(4482.5,134.85963888888887)(4590.5,138.915)(4698.5,143.019)(4806.5,147.14347222222221)(4914.5,151.3475)(5022.5,155.5595)(5130.5,159.81372222222222)(5238.5,164.133)(5346.5,168.45950925925928)(5454.5,172.8645)(5562.5,177.2925)(5670.5,181.76361111111112)(5778.5,186.299)(5886.5,190.85425925925927)(5994.5,195.4895)(6102.5,200.14036111111113)(6210.5,204.87)(6318.5,209.62602777777778)(6426.5,214.4575)(6534.5,219.3183425925926)(6642.5,224.245)(6750.5,229.213)(6858.5,234.2241111111111)(6966.5,239.30275)(7074.5,244.456)(7182.5,249.64916666666664)(7290.5,254.9225)(7398.5,260.2515833333333)(7506.5,265.6525555555556)(7614.5,271.1275)(7722.5,276.6578611111111)(7830.5,282.2715925925926)(7938.5,287.9685)(8046.5,293.747)(8154.5,299.62122222222223)(8262.5,305.59355555555555)(8370.5,311.66200925925926)(8478.5,317.83197222222225)(8586.5,324.1066666666667)(8694.5,330.4918611111111)(8802.5,336.99280555555555)(8910.5,343.61651851851855)(9018.5,350.41122222222225)(9126.5,357.3658148148148)(9234.5,364.50728703703703)(9342.5,371.844787037037)(9450.5,379.40802777777776)(9558.5,387.22408333333334)(9666.5,395.33983333333333)(9774.5,403.80940740740743)(9882.5,412.7536203703704)(9990.5,422.26373148148144)(10098.5,432.63098148148146)(10206.5,444.28489814814816)(10314.5,457.76581481481486)(10422.5,475.7292037037037)(10530.5,508.08969444444443)(10596,562.2875217391305)(10609,606.935)};
\addlegendentry{Z3Seq}

\addplot[name path=pathOSTRICH , colourOSTRICH, line width=1.5pt, dotted] coordinates {(1,0.741)(43.5,44.94746511627907)(129.5,146.9616511627907)(215.5,257.48620930232556)(301.5,372.4051976744186)(387.5,490.67495348837207)(473.5,611.8609883720931)(559.5,735.7402325581395)(645.5,862.2752674418605)(731.5,991.6876627906977)(817.5,1123.6380697674417)(903.5,1257.9013372093023)(989.5,1394.8169302325582)(1075.5,1533.942430232558)(1161.5,1674.9191162790698)(1247.5,1817.5090813953489)(1333.5,1961.8142325581396)(1419.5,2107.771930232558)(1505.5,2255.2819418604654)(1591.5,2404.4249534883725)(1677.5,2555.0977558139534)(1763.5,2707.2265581395345)(1849.5,2860.8377441860466)(1935.5,3015.9242558139536)(2021.5,3172.2990697674422)(2107.5,3330.130639534884)(2193.5,3489.6160465116277)(2279.5,3650.886209302326)(2365.5,3813.741372093023)(2451.5,3978.5878953488373)(2537.5,4145.205290697674)(2623.5,4313.637523255814)(2709.5,4484.107325581395)(2795.5,4656.646069767442)(2881.5,4831.097360465116)(2967.5,5007.5295)(3053.5,5185.771406976744)(3139.5,5365.600069767442)(3225.5,5547.0971744186045)(3311.5,5730.2754651162795)(3397.5,5914.922755813954)(3483.5,6100.775872093022)(3569.5,6287.998837209302)(3655.5,6476.696093023255)(3741.5,6666.811267441861)(3827.5,6858.177639534884)(3913.5,7050.711453488372)(3999.5,7244.355639534884)(4085.5,7439.031906976745)(4171.5,7634.867848837209)(4257.5,7831.893058139534)(4343.5,8029.896116279069)(4429.5,8228.874941860466)(4515.5,8428.930616279069)(4601.5,8630.054616279069)(4687.5,8832.318267441862)(4773.5,9035.793906976745)(4859.5,9240.507569767442)(4945.5,9446.523604651162)(5031.5,9653.765441860465)(5117.5,9862.005267441862)(5203.5,10071.376511627906)(5289.5,10282.060302325583)(5375.5,10493.92)(5461.5,10706.976941860465)(5547.5,10921.303267441861)(5633.5,11136.852674418604)(5719.5,11353.69668604651)(5805.5,11572.019627906977)(5891.5,11791.764674418604)(5977.5,12012.935313953489)(6063.5,12235.660453488372)(6149.5,12460.02308139535)(6235.5,12686.00434883721)(6321.5,12913.830569767442)(6407.5,13143.537906976744)(6493.5,13375.497883720931)(6579.5,13609.57084883721)(6665.5,13846.043569767442)(6751.5,14084.963825581395)(6837.5,14326.703569767442)(6923.5,14571.299476744187)(7009.5,14819.031441860465)(7095.5,15070.219244186046)(7181.5,15325.323418604652)(7267.5,15584.147604651162)(7353.5,15847.513360465116)(7439.5,16116.183058139535)(7525.5,16390.33041860465)(7611.5,16670.45136046512)(7697.5,16958.06230232558)(7783.5,17254.16811627907)(7869.5,17559.92151162791)(7955.5,17877.303081395352)(8041.5,18206.947360465118)(8127.5,18554.478593023254)(8213.5,18924.77562790698)(8299.5,19321.39338372093)(8385.5,19777.598372093023)(8435,20106.05769230769)(8443,20174.836)};
\addlegendentry{OSTRICH}

\addplot[name path=pathZ3-Trau , colourZ3-Trau, line width=1.5pt,dash pattern={on 7pt off 2pt on 1pt off 3pt}] coordinates {(1,0.02)(45.5,1.0565444444444445)(135.5,3.2051666666666665)(225.5,5.435666666666667)(315.5,7.736)(405.5,10.080511111111111)(495.5,12.4895)(585.5,14.921011111111111)(675.5,17.411)(765.5,19.931)(855.5,22.4734)(945.5,25.0795)(1035.5,27.6895)(1125.5,30.333733333333335)(1215.5,33.033)(1305.5,35.733)(1395.5,38.44553333333333)(1485.5,41.2255)(1575.5,44.0155)(1665.5,46.80561111111111)(1755.5,49.645)(1845.5,52.525)(1935.5,55.405)(2025.5,58.299166666666665)(2115.5,61.2605)(2205.5,64.2305)(2295.5,67.2024)(2385.5,70.234)(2475.5,73.294)(2565.5,76.35681111111111)(2655.5,79.4815)(2745.5,82.6315)(2835.5,85.7889)(2925.5,89.013)(3015.5,92.253)(3105.5,95.49401111111112)(3195.5,98.7915)(3285.5,102.1215)(3375.5,105.4585)(3465.5,108.862)(3555.5,112.282)(3645.5,115.72166666666668)(3735.5,119.2265)(3825.5,122.7365)(3915.5,126.25701111111111)(4005.5,129.845)(4095.5,133.445)(4185.5,137.0640111111111)(4275.5,140.7485)(4365.5,144.4385)(4455.5,148.1577)(4545.5,151.936)(4635.5,155.7215111111111)(4725.5,159.5725)(4815.5,163.4425)(4905.5,167.34501111111112)(4995.5,171.304)(5085.5,175.28236666666666)(5175.5,179.3265)(5265.5,183.37820000000002)(5355.5,187.489)(5445.5,191.629)(5535.5,195.79283333333333)(5625.5,200.0195)(5715.5,204.25573333333332)(5805.5,208.558)(5895.5,212.8816111111111)(5985.5,217.2685)(6075.5,221.68853333333334)(6165.5,226.176)(6255.5,230.68466666666666)(6345.5,235.2605)(6435.5,239.87361111111113)(6525.5,244.55)(6615.5,249.25456666666668)(6705.5,254.0270111111111)(6795.5,258.868)(6885.5,263.75561111111114)(6975.5,268.7080111111111)(7065.5,273.727)(7155.5,278.8105111111111)(7245.5,283.9643333333333)(7335.5,289.1865111111111)(7425.5,294.47773333333333)(7515.5,299.839)(7605.5,305.2828666666666)(7695.5,310.83231111111115)(7785.5,316.4719)(7875.5,322.2251111111111)(7965.5,328.08116666666666)(8055.5,334.0618444444445)(8145.5,340.1675333333333)(8235.5,346.4065111111111)(8325.5,352.8138333333333)(8415.5,359.43440000000004)(8505.5,366.3056666666667)(8595.5,373.5592888888889)(8685.5,381.5135777777778)(8754,389.2892340425532)(8779,400.461)};
\addlegendentry{Z3-Trau}

\addplot[name path=pathZ3str3 , colourZ3str3, line width=1.5pt,dash dot] coordinates {(1,0.012)(47.5,0.7518617021276596)(141.5,2.4066063829787234)(235.5,4.1785)(329.5,5.995627659574469)(423.5,7.874)(517.5,9.754)(611.5,11.656127659574468)(705.5,13.6255)(799.5,15.5995)(893.5,17.5735)(987.5,19.554202127659572)(1081.5,21.604)(1175.5,23.672)(1269.5,25.74)(1363.5,27.808)(1457.5,29.876)(1551.5,31.96477659574468)(1645.5,34.1215)(1739.5,36.2835)(1833.5,38.4455)(1927.5,40.6075)(2021.5,42.76953191489362)(2115.5,44.981)(2209.5,47.237)(2303.5,49.493)(2397.5,51.749)(2491.5,54.005)(2585.5,56.263936170212766)(2679.5,58.5875)(2773.5,60.9375)(2867.5,63.2875)(2961.5,65.6375)(3055.5,67.99244680851064)(3149.5,70.415)(3243.5,72.859)(3337.5,75.303)(3431.5,77.76003191489362)(3525.5,80.2875)(3619.5,82.8255)(3713.5,85.3635)(3807.5,87.94603191489361)(3901.5,90.578)(3995.5,93.21127659574468)(4089.5,95.9045)(4183.5,98.6305)(4277.5,101.38369148936171)(4371.5,104.201)(4465.5,107.0288829787234)(4559.5,109.9265)(4653.5,112.8405)(4747.5,115.802)(4841.5,118.81)(4935.5,121.83877659574468)(5029.5,124.9355)(5123.5,128.04123404255319)(5217.5,131.213)(5311.5,134.4121914893617)(5405.5,137.6765)(5499.5,140.96794680851065)(5593.5,144.32)(5687.5,147.7040106382979)(5781.5,151.1365)(5875.5,154.6176914893617)(5969.5,158.164)(6063.5,161.73945744680853)(6157.5,165.3805)(6251.5,169.05282978723406)(6345.5,172.794)(6439.5,176.5577340425532)(6533.5,180.3875)(6627.5,184.2561595744681)(6721.5,188.195)(6815.5,192.17832978723405)(6909.5,196.21979787234042)(7003.5,200.316)(7097.5,204.46556382978721)(7191.5,208.6855638297872)(7285.5,212.966)(7379.5,217.3025106382979)(7473.5,221.7095)(7567.5,226.1691595744681)(7661.5,230.71131914893616)(7755.5,235.3362765957447)(7849.5,240.04797872340427)(7943.5,244.85527659574467)(8037.5,249.7543829787234)(8131.5,254.7534574468085)(8225.5,259.8639787234043)(8319.5,265.07851063829787)(8413.5,270.41279787234043)(8507.5,275.88305319148935)(8601.5,281.50894680851064)(8695.5,287.32851063829787)(8789.5,293.42238297872336)(8883.5,299.8523404255319)(8977.5,306.7179042553191)(9071.5,314.4735106382979)(9165.5,375.96380851063833)(9222.5,698.37135)(9234,828.245)};
\addlegendentry{Z3str3}

\addplot[name path=pathZ3str3RE , colourZ3str3RE, line width=1.5pt] coordinates {(1,0.008)(55,0.612302752293578)(164,1.9329908256880732)(273,3.352798165137615)(382,4.856)(491,6.382192660550459)(600,7.969)(709,9.604)(818,11.241119266055046)(927,12.95)(1036,14.694)(1145,16.438)(1254,18.2274128440367)(1363,20.08)(1472,21.933)(1581,23.797238532110093)(1690,25.743)(1799,27.705)(1908,29.667)(2017,31.64867889908257)(2126,33.711)(2235,35.782)(2344,37.853)(2453,39.964963302752295)(2562,42.144)(2671,44.324)(2780,46.504)(2889,48.70308256880734)(2998,50.983)(3107,53.272)(3216,55.561)(3325,57.85545871559633)(3434,60.228)(3543,62.626)(3652,65.024)(3761,67.42432110091742)(3870,69.897)(3979,72.404)(4088,74.911)(4197,77.4258990825688)(4306,80.021)(4415,82.637)(4524,85.253)(4633,87.8735504587156)(4742,90.571)(4851,93.296)(4960,96.021)(5069,98.74756880733946)(5178,101.544)(5287,104.378)(5396,107.212)(5505,110.04619266055046)(5614,112.941)(5723,115.884)(5832,118.827)(5941,121.77298165137614)(6050,124.793)(6159,127.845)(6268,130.897)(6377,133.97115596330278)(6486,137.125)(6595,140.286)(6704,143.4474128440367)(6813,146.672)(6922,149.942)(7031,153.212)(7140,156.53209174311925)(7249,159.911)(7358,163.29)(7467,166.71171559633027)(7576,170.199)(7685,173.687)(7794,177.229)(7903,180.826)(8012,184.4337889908257)(8121,188.123)(8230,191.82960550458714)(8339,195.601)(8448,199.416)(8557,203.26774311926607)(8666,207.19)(8775,211.168)(8884,215.20240366972476)(8993,219.306)(9102,223.45923853211008)(9211,227.694)(9320,231.97045871559632)(9429,236.33014678899082)(9538,240.77699082568805)(9647,245.326)(9756,249.94941284403671)(9865,254.6914128440367)(9974,259.5555504587156)(10083,264.55191743119263)(10192,269.7016422018349)(10301,275.0770825688073)(10410,280.72149541284404)(10519,286.8079633027523)(10604,292.33534426229505)(10636,295.067)};
\addlegendentry{Z3str3RE}

\path[name path=axisCVC4] (axis cs:0,0) -- (axis cs:10557,0);
\addplot [thick,color=colourCVC4,fill=colourCVC4,fill opacity=0.1] fill between [of=pathCVC4 and axisCVC4];
\path[name path=axisZ3Seq] (axis cs:0,0) -- (axis cs:10609,0);
\addplot [thick,color=colourZ3Seq,fill=colourZ3Seq,fill opacity=0.1] fill between [of=pathZ3Seq and axisZ3Seq];
\path[name path=axisOSTRICH] (axis cs:0,0) -- (axis cs:8443,0);
\addplot [thick,color=colourOSTRICH,fill=colourOSTRICH,fill opacity=0.1] fill between [of=pathOSTRICH and axisOSTRICH];
\path[name path=axisZ3-Trau] (axis cs:0,0) -- (axis cs:8779,0);
\addplot [thick,color=colourZ3-Trau,fill=colourZ3-Trau,fill opacity=0.1] fill between [of=pathZ3-Trau and axisZ3-Trau];
\path[name path=axisZ3str3] (axis cs:0,0) -- (axis cs:9234,0);
\addplot [thick,color=colourZ3str3,fill=colourZ3str3,fill opacity=0.1] fill between [of=pathZ3str3 and axisZ3str3];
\path[name path=axisZ3str3RE] (axis cs:0,0) -- (axis cs:10636,0);
\addplot [thick,color=colourZ3str3RE,fill=colourZ3str3RE,fill opacity=0.1] fill between [of=pathZ3str3RE and axisZ3str3RE];
\end{axis}\end{tikzpicture}}

\caption{Cactus plot showing detailed results for the \benchmarkxform{} benchmark.}
\label{fig:cactus_transformed}

\end{figure}
\begin{table}[t!]

\resizebox{0.95\textwidth}{!}{
\begin{tabular}{|c |c |c |c |c |c |c |}
\hline
&CVC4&Z3Seq&OSTRICH&Z3-Trau&Z3str3&Z3str3RE\\ 
  \hline\hline 
sat &4541&\textbf{4633}&3899&3672&4417&4599\\ 
 \hline
unsat &6016&5976&4549&\textbf{6282}&4817&6037\\ 
 \hline
\hline 
 unknown &\textbf{0}&\textbf{0}&2233&721&\textbf{0}&6\\ 
 \hline
timeout &125&73&\textbf{1}&7&1448&40\\ 
 \hline
soundness error &\textbf{0}&\textbf{0}&5&1241&\textbf{0}&\textbf{0}\\ 
 \hline
program crashes &\textbf{0}&\textbf{0}&\textbf{0}&718&\textbf{0}&\textbf{0}\\ 
 \hline
\hline 
 Total correct &10557&10609&8443&8713&9234&\textbf{10636}\\ 
 \hline
Contribution score &0.5&0.0&--&--&0.0&\textbf{4.83}\\ 
\hline
Time (s) &2969.643&2066.935&23094.737&\textbf{722.545}&29788.245&1095.209\\ 
 \hline
Time w/o timeouts (s) &469.643&606.935&23074.737&582.545&828.245&\textbf{295.209}\\ 
 \hline
\end{tabular}}
\vspace{0.1cm}
\caption{Detailed results for the \benchmarkxform{} benchmark. \toolname{} has the biggest lead with a score of 1.0.}
\label{tab:cactus_transformed}

\end{table}

%% file: data4.tex
\begin{figure}[t!]

\resizebox{.95\textwidth}{!}{\pgfplotsset{scaled x ticks=false}\pgfplotsset{scaled y ticks=false}\begin{tikzpicture}\begin{axis}[title=RegEx Collected,xmin=-1000,xlabel=Solved instances,ylabel=Time (seconds),,legend columns=2,legend style={nodes={scale=1, transform shape}, fill=none,anchor=east,align=center },axis line style={draw=none}, xtick pos=left, ytick pos=left, ymajorgrids=true, legend style={draw=none,fill=white},x post scale=2,y post scale=1.25,ymax=12500]
\addplot[name path=pathCVC4 , colourCVC4, line width=1.5pt,dashed] coordinates {(1,0.009)(114,1.6043348017621144)(341,5.333255506607929)(568,9.490105726872248)(795,13.936559471365639)(1022,18.63684140969163)(1249,23.56359911894273)(1476,28.687)(1703,34.00011453744494)(1930,39.50094713656387)(2157,45.19337444933921)(2384,51.06932599118943)(2611,57.115)(2838,63.31030396475771)(3065,69.67125991189427)(3292,76.201)(3519,82.84074008810572)(3746,89.64639647577091)(3973,96.614)(4200,103.6783832599119)(4427,110.913)(4654,118.25777533039647)(4881,125.746)(5108,133.35)(5335,141.07276211453745)(5562,148.946)(5789,156.9169295154185)(6016,165.058)(6243,173.26808810572686)(6470,181.647)(6697,190.11618061674008)(6924,198.7372422907489)(7151,207.487)(7378,216.35222466960352)(7605,225.381)(7832,234.49140969162997)(8059,243.772)(8286,253.15316740088107)(8513,262.6830440528634)(8740,272.335)(8967,282.11484581497797)(9194,292.063)(9421,302.1219691629956)(9648,312.3322907488987)(9875,322.672)(10102,333.1658986784141)(10329,343.83555506607934)(10556,354.681)(10783,365.6479691629956)(11010,376.78286343612336)(11237,388.0900660792952)(11464,399.559)(11691,411.2007841409692)(11918,423.05125991189425)(12145,435.10514096916296)(12372,447.3725154185022)(12599,459.87835682819383)(12826,472.60969162995593)(13053,485.55589427312776)(13280,498.74162995594713)(13507,512.1870704845816)(13734,525.8877885462556)(13961,539.8567400881058)(14188,554.1396696035242)(14415,568.7316563876652)(14642,583.6663744493392)(14869,598.9352643171807)(15096,614.5943436123348)(15323,630.6912026431718)(15550,647.2369162995595)(15777,664.2850881057269)(16004,681.9400220264317)(16231,700.1907356828194)(16458,719.0646123348018)(16685,738.6549251101322)(16912,758.9736607929516)(17139,780.1746255506607)(17366,802.3928414096916)(17593,825.9486828193832)(17820,851.1296343612335)(18047,878.1704449339207)(18274,907.5060748898678)(18501,939.7701894273129)(18728,974.897691629956)(18955,1012.9952466960352)(19182,1054.229744493392)(19409,1099.5227092511013)(19636,1150.8639339207048)(19863,1211.4973876651982)(20090,1283.5564185022026)(20317,1383.7662863436121)(20544,1517.0738898678414)(20771,1666.339101321586)(20998,1829.17595154185)(21225,2032.4169603524228)(21452,2861.2594361233478)(21679,4307.138356828194)(21906,6183.727528634361)(22115,8934.336246073299)(22212,10350.224)};
\addlegendentry{CVC4}

\addplot[name path=pathZ3Seq , colourZ3Seq, line width=1.5pt, densely dotted] coordinates {(1,0.017)(103.5,2.264344660194175)(309.5,7.088412621359224)(515.5,12.209941747572815)(721.5,17.55560194174757)(927.5,23.092679611650485)(1133.5,28.795)(1339.5,34.63355825242718)(1545.5,40.62757281553398)(1751.5,46.779)(1957.5,53.03633495145631)(2163.5,59.42116019417476)(2369.5,65.926)(2575.5,72.52659223300971)(2781.5,79.2725)(2987.5,86.0903786407767)(3193.5,93.062)(3399.5,100.10861165048544)(3605.5,107.3055)(3811.5,114.55495631067961)(4017.5,121.956)(4223.5,129.38620388349514)(4429.5,136.9675)(4635.5,144.59273300970872)(4841.5,152.351)(5047.5,160.179)(5253.5,168.07673300970873)(5459.5,176.1075)(5665.5,184.17276699029128)(5871.5,192.392)(6077.5,200.6365922330097)(6283.5,209.0185)(6489.5,217.4645)(6695.5,225.97147572815535)(6901.5,234.618)(7107.5,243.28733009708736)(7313.5,252.1095)(7519.5,260.9675)(7725.5,269.9154417475728)(7931.5,278.979)(8137.5,288.1001893203884)(8343.5,297.3635)(8549.5,306.6469708737864)(8755.5,316.081)(8961.5,325.557)(9167.5,335.1238786407767)(9373.5,344.8055)(9579.5,354.53678640776695)(9785.5,364.415)(9991.5,374.34001941747573)(10197.5,384.4175)(10403.5,394.54500970873784)(10609.5,404.826)(10815.5,415.15671844660193)(11021.5,425.6415)(11227.5,436.1682669902913)(11433.5,446.849)(11639.5,457.59983980582524)(11845.5,468.5025)(12051.5,479.4885970873786)(12257.5,490.6091019417476)(12463.5,501.8425)(12669.5,513.1906601941747)(12875.5,524.692)(13081.5,536.3062038834952)(13287.5,548.0526893203884)(13493.5,559.942)(13699.5,571.9464466019417)(13905.5,584.1039126213592)(14111.5,596.4162669902913)(14317.5,608.8895)(14523.5,621.5040970873786)(14729.5,634.2845825242719)(14935.5,647.2427038834952)(15141.5,660.3973883495146)(15347.5,673.7641019417476)(15553.5,687.3662427184466)(15759.5,701.2004805825243)(15965.5,715.2624805825243)(16171.5,729.6037669902913)(16377.5,744.2724514563107)(16583.5,759.2605)(16789.5,774.5760533980583)(16995.5,790.2594271844661)(17201.5,806.3397524271845)(17407.5,822.8788495145631)(17613.5,839.9413155339806)(17819.5,857.7020194174756)(18025.5,876.4284805825243)(18231.5,896.4705097087378)(18437.5,918.8223737864078)(18643.5,946.6145776699029)(18849.5,981.0990825242718)(19055.5,1019.8828980582524)(19261.5,1062.3955436893202)(19467.5,1109.472572815534)(19673.5,1165.2795145631067)(19879.5,1252.756572815534)(20070.5,1486.6743920454546)(20160,1993.484)};
\addlegendentry{Z3Seq}

\addplot[name path=pathOSTRICH , colourOSTRICH, line width=1.5pt, dotted] coordinates {(1,0.963)(70,95.95328057553957)(209,317.5416330935252)(348,558.2403669064748)(487,807.0735323741008)(626,1061.7848273381296)(765,1320.9988705035971)(904,1583.9980719424461)(1043,1850.313107913669)(1182,2119.4639496402874)(1321,2391.081309352518)(1460,2665.132676258993)(1599,2941.562143884892)(1738,3220.0257913669066)(1877,3500.438654676259)(2016,3782.8274604316543)(2155,4066.871762589928)(2294,4352.527611510792)(2433,4639.561892086331)(2572,4928.025064748202)(2711,5218.137230215827)(2850,5509.955618705036)(2989,5803.4522230215825)(3128,6098.407683453238)(3267,6394.939179856115)(3406,6693.154683453237)(3545,6993.038381294964)(3684,7294.60390647482)(3823,7597.650071942446)(3962,7902.217467625899)(4101,8208.322928057554)(4240,8516.003287769783)(4379,8825.339122302157)(4518,9136.520949640288)(4657,9449.287647482015)(4796,9763.798683453238)(4935,10080.263582733814)(5074,10398.514791366906)(5213,10718.48792086331)(5352,11040.249784172662)(5491,11363.872978417266)(5630,11689.546057553956)(5769,12017.450007194246)(5908,12347.792316546762)(6047,12680.60456115108)(6186,13016.34720143885)(6325,13355.228007194246)(6464,13697.062316546762)(6603,14042.445014388488)(6742,14391.950726618705)(6881,14746.046215827339)(7020,15105.197007194245)(7159,15469.613553956835)(7298,15839.61420143885)(7437,16215.978489208634)(7576,16599.61313669065)(7715,16990.84354676259)(7854,17388.615035971223)(7993,17791.553726618702)(8132,18199.8160647482)(8271,18613.770791366904)(8410,19033.166517985614)(8549,19457.74643884892)(8688,19887.228582733813)(8827,20321.42998561151)(8966,20760.531201438847)(9105,21204.715884892088)(9244,21653.977482014387)(9383,22108.476805755396)(9522,22569.0895971223)(9661,23036.884791366905)(9800,23512.15438848921)(9939,23995.244532374098)(10078,24486.649402877698)(10217,24986.187071942444)(10356,25493.466791366904)(10495,26008.347381294967)(10634,26531.385474820145)(10773,27062.61843884892)(10912,27602.335136690646)(11051,28149.832920863308)(11190,28704.555971223024)(11329,29267.476309352518)(11468,29838.942035971224)(11607,30418.584762589926)(11746,31006.01256115108)(11885,31602.551496402877)(12024,32210.101381294964)(12163,32828.29960431655)(12302,33458.97757553957)(12441,34103.57407194245)(12580,34763.35728057554)(12719,35440.604517985616)(12858,36140.66408633094)(12997,36866.012892086335)(13136,37615.940669064745)(13275,38394.040992805756)(13414,39240.46368345324)(13553,40415.1047913669)(13632,41390.67331578948)(13643,41576.524)};
\addlegendentry{OSTRICH}

\addplot[name path=pathZ3-Trau , colourZ3-Trau, line width=1.5pt,dash pattern={on 7pt off 2pt on 1pt off 3pt}] coordinates {(1,0.021)(93,2.3040540540540544)(278,7.1193243243243245)(463,12.159351351351352)(648,17.345243243243242)(833,22.671)(1018,28.108248648648647)(1203,33.670054054054056)(1388,39.369)(1573,45.13881621621622)(1758,51.045)(1943,57.042648648648644)(2128,63.151881081081086)(2313,69.387)(2498,75.72735675675675)(2683,82.1996)(2868,88.8)(3053,95.49543243243242)(3238,102.32703243243243)(3423,109.268)(3608,116.3078918918919)(3793,123.481)(3978,130.72551351351353)(4163,138.108)(4348,145.58024864864865)(4533,153.16604324324325)(4718,160.869)(4903,168.64762702702703)(5088,176.558)(5273,184.54481621621622)(5458,192.669)(5643,200.84023243243243)(5828,209.149)(6013,217.52362162162163)(6198,226.027)(6383,234.60575675675673)(6568,243.30484324324323)(6753,252.133)(6938,261.04423243243247)(7123,270.093)(7308,279.23564864864863)(7493,288.4958972972973)(7678,297.891)(7863,307.3639513513513)(8048,316.97211351351353)(8233,326.691)(8418,336.5238432432432)(8603,346.49751351351347)(8788,356.608)(8973,366.8615675675676)(9158,377.2393945945946)(9343,387.756)(9528,398.3861513513514)(9713,409.15524324324326)(9898,420.0728162162162)(10083,431.13902702702705)(10268,442.35)(10453,453.72600540540543)(10638,465.25243243243244)(10823,476.94902702702706)(11008,488.8051351351352)(11193,500.8363945945946)(11378,513.0653135135135)(11563,525.510854054054)(11748,538.1667621621621)(11933,551.040027027027)(12118,564.1506486486486)(12303,577.5008810810812)(12488,591.1034648648648)(12673,605.0039891891892)(12858,619.2891351351351)(13043,633.9784702702702)(13228,649.1127891891892)(13413,664.8049297297298)(13598,681.1325027027026)(13783,698.1714432432433)(13968,715.9630324324324)(14153,734.7967945945946)(14338,754.8878378378379)(14523,776.0934702702702)(14708,798.4790432432433)(14893,822.0978432432432)(15078,846.7652162162162)(15263,872.6489513513513)(15448,900.0625621621622)(15633,928.9504054054055)(15818,959.5314972972974)(16003,992.251772972973)(16188,1027.3043027027027)(16373,1065.1179135135135)(16558,1106.5884432432433)(16743,1152.762864864865)(16928,1206.119945945946)(17113,1276.6789567567566)(17298,1366.6714054054055)(17483,1525.4639945945946)(17668,1868.4734270270271)(17853,2550.809054054054)(18038,3784.824291891892)(18132.5,5168.105)(18136,5217.583)};
\addlegendentry{Z3-Trau}

\addplot[name path=pathZ3str3 , colourZ3str3, line width=1.5pt,dash dot] coordinates {(1,0.017)(103,2.2264536585365855)(308,6.873063414634147)(513,11.740834146341463)(718,16.781)(923,21.937975609756098)(1128,27.248)(1333,32.64240487804878)(1538,38.173)(1743,43.778902439024385)(1948,49.516)(2153,55.31421951219512)(2358,61.253)(2563,67.27141951219512)(2768,73.4190731707317)(2973,79.677)(3178,86.03251219512195)(3383,92.504)(3588,99.06451219512195)(3793,105.741)(3998,112.50651219512194)(4203,119.388)(4408,126.35874634146342)(4613,133.448)(4818,140.623)(5023,147.89900487804877)(5228,155.279)(5433,162.7420243902439)(5638,170.326)(5843,178.0090487804878)(6048,185.799)(6253,193.66241951219513)(6458,201.6551756097561)(6663,209.761)(6868,217.9631219512195)(7073,226.293)(7278,234.69923414634147)(7483,243.228)(7688,251.8381024390244)(7893,260.557)(8098,269.3739804878049)(8303,278.318)(8508,287.34373658536583)(8713,296.509)(8918,305.7637804878049)(9123,315.172)(9328,324.65363414634146)(9533,334.28)(9738,344.01110243902434)(9943,353.85540487804883)(10148,363.836)(10353,373.9170048780488)(10558,384.15)(10763,394.4839268292683)(10968,404.95541463414634)(11173,415.58)(11378,426.3341756097561)(11583,437.22718536585364)(11788,448.27571219512197)(11993,459.473)(12198,470.81890243902444)(12403,482.351243902439)(12608,494.06559999999996)(12813,505.95897560975607)(13018,518.0411024390244)(13223,530.3336243902439)(13428,542.8545365853658)(13633,555.5852487804879)(13838,568.5281121951219)(14043,581.7019170731708)(14248,595.146243902439)(14453,608.8557804878049)(14658,622.8558)(14863,637.1969804878048)(15068,651.8989560975609)(15273,666.9906390243902)(15478,682.5244536585366)(15683,698.5892780487806)(15888,715.3176731707317)(16093,732.8633804878049)(16298,751.598312195122)(16503,771.8741902439024)(16708,794.251712195122)(16913,819.2873024390244)(17118,847.4867463414635)(17323,877.7979170731708)(17528,909.5736487804878)(17733,943.5206243902439)(17938,981.4391902439024)(18143,1024.7791268292683)(18348,1074.756092682927)(18553,1131.0165219512196)(18758,1194.3116439024388)(18963,1267.343)(19168,1350.8028097560975)(19373,1447.4682878048782)(19578,1563.6074536585365)(19783,1787.7796975609756)(19985.5,2316.803195)(20087,2816.004)};
\addlegendentry{Z3str3}

\addplot[name path=pathZ3str3RE , colourZ3str3RE, line width=1.5pt] coordinates {(1,0.007)(110.5,1.1454863636363637)(330.5,3.7729681818181815)(550.5,6.617545454545455)(770.5,9.592)(990.5,12.672)(1210.5,15.752)(1430.5,18.8498)(1650.5,22.1105)(1870.5,25.4105)(2090.5,28.7105)(2310.5,32.0105)(2530.5,35.3105)(2750.5,38.6105)(2970.5,42.008354545454544)(3190.5,45.528)(3410.5,49.048)(3630.5,52.568)(3850.5,56.088)(4070.5,59.608)(4290.5,63.128)(4510.5,66.648)(4730.5,70.168)(4950.5,73.78304545454546)(5170.5,77.5225)(5390.5,81.2625)(5610.5,85.0025)(5830.5,88.7425)(6050.5,92.4825)(6270.5,96.2225)(6490.5,99.9625)(6710.5,103.7025)(6930.5,107.4425)(7150.5,111.22736363636363)(7370.5,115.173)(7590.5,119.133)(7810.5,123.093)(8030.5,127.053)(8250.5,131.013)(8470.5,134.973)(8690.5,138.933)(8910.5,142.893)(9130.5,146.85471818181816)(9350.5,150.9505)(9570.5,155.1305)(9790.5,159.3105)(10010.5,163.4905)(10230.5,167.6705)(10450.5,171.8505)(10670.5,176.0305)(10890.5,180.30186363636363)(11110.5,184.701)(11330.5,189.101)(11550.5,193.501)(11770.5,197.9017772727273)(11990.5,202.4295)(12210.5,207.0495)(12430.5,211.6695)(12650.5,216.29506818181818)(12870.5,221.069)(13090.5,225.909)(13310.5,230.79515)(13530.5,235.8415)(13750.5,240.95575)(13970.5,246.226)(14190.5,251.59827727272727)(14410.5,257.11910454545455)(14630.5,262.8130454545455)(14850.5,268.7088636363637)(15070.5,274.86572727272727)(15290.5,281.37265454545457)(15510.5,288.4409727272727)(15730.5,296.2691)(15950.5,304.8845227272727)(16170.5,314.3141181818182)(16390.5,324.52733181818184)(16610.5,335.75216818181815)(16830.5,348.64034545454547)(17050.5,363.10413636363637)(17270.5,378.6765363636364)(17490.5,395.32815909090914)(17710.5,413.16744545454543)(17930.5,432.6688409090909)(18150.5,454.36280454545454)(18370.5,479.00889545454544)(18590.5,506.9938181818182)(18810.5,536.6671590909091)(19030.5,567.3577909090909)(19250.5,598.9474727272727)(19470.5,631.7009909090909)(19690.5,666.5226909090909)(19910.5,707.1278000000001)(20130.5,759.9534727272728)(20350.5,818.7017454545455)(20570.5,882.8866954545455)(20790.5,952.0425954545454)(21010.5,1031.3764954545454)(21230.5,1136.5260454545455)(21450.5,1277.1239772727272)(21581,1431.705243902439)(21603,1504.149)};
\addlegendentry{Z3str3RE}

\path[name path=axisCVC4] (axis cs:0,0) -- (axis cs:22212,0);
\addplot [thick,color=colourCVC4,fill=colourCVC4,fill opacity=0.1] fill between [of=pathCVC4 and axisCVC4];
\path[name path=axisZ3Seq] (axis cs:0,0) -- (axis cs:20160,0);
\addplot [thick,color=colourZ3Seq,fill=colourZ3Seq,fill opacity=0.1] fill between [of=pathZ3Seq and axisZ3Seq];
\path[name path=axisOSTRICH] (axis cs:0,0) -- (axis cs:13643,0);
\addplot [thick,color=colourOSTRICH,fill=colourOSTRICH,fill opacity=0.1] fill between [of=pathOSTRICH and axisOSTRICH];
\path[name path=axisZ3-Trau] (axis cs:0,0) -- (axis cs:18136,0);
\addplot [thick,color=colourZ3-Trau,fill=colourZ3-Trau,fill opacity=0.1] fill between [of=pathZ3-Trau and axisZ3-Trau];
\path[name path=axisZ3str3] (axis cs:0,0) -- (axis cs:20087,0);
\addplot [thick,color=colourZ3str3,fill=colourZ3str3,fill opacity=0.1] fill between [of=pathZ3str3 and axisZ3str3];
\path[name path=axisZ3str3RE] (axis cs:0,0) -- (axis cs:21603,0);
\addplot [thick,color=colourZ3str3RE,fill=colourZ3str3RE,fill opacity=0.1] fill between [of=pathZ3str3RE and axisZ3str3RE];
\end{axis}\end{tikzpicture}}

\caption{Cactus plot showing detailed performance for the \benchmarkcollected{} benchmark.}
\label{fig:cactus_collected}

\end{figure}
\begin{table}[ht!]

\resizebox{0.95\textwidth}{!}{
\begin{tabular}{|c |c |c |c |c |c |c |}
\hline
&CVC4&Z3Seq&OSTRICH&Z3-Trau&Z3str3&Z3str3RE\\ 
  \hline\hline 
sat &\textbf{12077}&10712&5134&10714&10768&11553\\ 
 \hline
unsat &\textbf{10135}&9448&8532&10115&9332&10050\\ 
 \hline
\hline 
 unknown &\textbf{0}&\textbf{0}&8652&546&758&285\\ 
 \hline
timeout &213&2265&\textbf{107}&1050&1567&537\\ 
 \hline
soundness error &\textbf{0}&\textbf{0}&23&2776&13&\textbf{0}\\ 
 \hline
program crashes &\textbf{0}&\textbf{0}&\textbf{0}&504&\textbf{0}&\textbf{0}\\ 
 \hline
\hline 
 Total correct &\textbf{22212}&20160&13643&18053&20087&21603\\ 
 \hline
Contribution score &\textbf{91.06}&3.51&--&--&--&14.54\\ 
\hline
Time (s) &14610.224&47293.484&71666.750&32220.939&35053.106&\textbf{13202.451}\\ 
 \hline
Time w/o timeouts (s) &10350.224&\textbf{1993.484}&69526.750&11220.939&3713.106&2462.451\\ 
 \hline
\end{tabular}
}
\vspace{0.1cm}
\caption{Detailed results for the \benchmarkcollected{} benchmark. CVC4 has the biggest lead with a score of 1.03.}
\label{tab:cactus_collected}

\end{table}

%% file: related-work.tex
\noindent{\bf Comparison with Z3str3:}
Z3str3~\cite{Z3str3} supports regex constraints via (incomplete) reduction to word equations. We have replaced this word-based technique with our automata-based approach introduced in this paper. As demonstrated by our evaluation, the length-aware automata-based approach used in \toolname{} is more efficient at solving these constraints,
and is sound and complete for the QF theory $T_{LRE}$.


\noindent{\bf Comparison with Z3's sequence solver:} Z3's sequence solver~\cite{z3}
supports a more general theory of ``sequences'' over arbitrary datatypes, which allows it to be used as a string solver.
Z3seq uses regular expression derivatives to reduce regex constraints without constructing automata.
The experiments show \toolname{} performs better than Z3seq overall.


\noindent{\bf Comparison with CVC4:} The CVC4 solver~\cite{CVC4-CAV14}
uses an algebraic approach to solving regex constraints. As shown in the experiments, \toolname{} performs better than CVC4, widely considered as one of the best SMT solvers for strings as well as many other theories.


\noindent \textbf{Comparison with Z3-Trau:} The Z3-Trau~\cite{z3-trau} solver builds on Trau~\cite{AbdullaFlattenAndConquer}, re-implemented in Z3, and enriched with new ideas e.g. a more efficient handling of string-number conversion. The evaluation of Z3-Trau exposed 5325 soundness errors and 2477 crashes on our benchmarks.


\noindent \textbf{Comparison with OSTRICH:} The OSTRICH solver~\cite{ostrich} implements a reduction from straight-line and acyclic fragments of an input formula to the emptiness problem of alternating finite automata. OSTRICH produced 10901 ``unknown'' responses and 4575 timeouts on our benchmarks, as well as 28 soundness errors.\looseness=-1


\noindent{\bf Related Algorithms and Theoretical Results:} The theory of word equations and various extensions have been studied extensively for many decades.  In 1977, Makanin proved that satisfiability for the QF theory of word equations is decidable~\cite{makanin}; in 1999, Plandowski showed that this is in PSPACE~\cite{plandowski99,plandowski2006}. Schulz~\cite{schulz} extended Makanin's algorithm to word equations with regex constraints. The satisfiability problem for the theory of word equations with length constraints still remains open~\cite{makanin,plandowski2006,ganesh2012,Matiyasevich}, although the status of many other extensions of this theory was clarified \cite{rp2018strings}. Automata-based approaches were used to reason about string constraints enhanced with a ReplaceAll function \cite{ChenCHLW18} or transducers \cite{sloth}.  

Liang et al.~\cite{CVC4-FROCOS15} present a formal calculus for a theory that extends $T_{LRE}$ with string concatenation (but not word equations). However, in that paper the authors do not present experimental results regarding implementation of the string calculus proposed. We have implemented an algorithm based on fundamentals of the theory and standard automata-based constructions, and presented a thorough experimental evaluation of our implementation, concluding with the presentation of a string solver using it.

Abdulla et al.~\cite{norn} present an automata-based solver called Norn built upon results involving construction of length constraints from regex constraints. This approach differs significantly from our method. In particular, Norn only uses automata in inferring length constraints implied by regular expressions, then uses an algebraic approach to solve the remainder of the formula. By contrast, our tool uses a hybrid approach that includes both algebraic solving and automata-based reasoning in a symbiotic loop. In addition, we present several novel heuristics using length information to guide the search and, in some cases, avoid constructing automata or computing intersections.

The prefix/suffix over-approximation heuristic is inspired partly by the work of Brzozowski on regex derivatives~\cite{brzozowski}. The heuristic we introduce is conceptually different as we examine possible prefixes (and suffixes) of strings that could be accepted by a regex in order to demonstrate unsatisfiability, rather than examining the set of all possible suffixes given a fixed prefix in order to demonstrate satisfiability. Our heuristic computes suffixes as well, whereas Brzozowski derivatives are traditionally computed with respect to prefixes of a string.
Newer versions of Z3seq, including the one we evaluated, use an algorithm based on symbolic derivatives to reason about regular expressions~\cite{stanford2020symbolic}.

%% file: conclusion.tex
\label{sec:conclusion}

In this paper, we have shown the power of length-aware and prefix/suffix reasoning for regex constraints with our algorithm and its implementation in \toolname{} via an extensive empirical comparison against five other state-of-the-art solvers (namely, CVC4, Z3seq, Z3str3, Z3-Trau, and OSTRICH) over a large and diverse benchmark of \totalinstances{} instances. Our length-aware method is very general and has wide applicability in the broad context of string solving. In the future, we plan to explore further length-aware heuristics which include more expressive functions and predicates, including \texttt{indexof}, \texttt{substr}, and string-number conversion. 

%% file: appendix_theory.tex
\newpage

\section*{Appendix}

\begin{restatable}{theorem}{decfour} \label{thm:undecWE}
The satisfiability problem for $T_{LRE,n,c}$ is undecidable.
\end{restatable}

\begin{proof}

We begin by considering the theory $T_{RE,n,c}$.

We will define a predicate $eqLen(\alpha,\beta)$, where $\alpha$ and $\beta$ are string terms, whose semantics is defined as follows: $eqLen(\alpha,\beta)$ is true if and only if the length of $\alpha$ equals the length of $\beta$ (i.e., $len(\alpha)=len(beta)$). 

We can express $eqLen(\alpha,\beta)$ as:\\
$eqLen(\alpha,\beta) = (z\in 1\{0\}^*) \land numstr(i,z) \land numstr (j,z0) \land numstr(n_a,1\alpha) \land numstr(n_b,1\beta) \land 
(i\leq n_a) \land (n_a +1 \leq j) \land (i\leq n_b) \land (n_b +1 \leq j)$. 

Indeed, we have $i=2^{len(z)}$ and $j=2^{len(z)+1}$. Then, we have $n_a=2^{len(\alpha)}+A$ and $n_b=2^{len(\beta)}+B$, where $numstr(A,\alpha)$ and $numstr(B,\beta)$ are true. Therefore, $2^{len(z)}\leq 2^{len(\alpha)}+A<2^{len(z+1)}$ and $2^{len(z)}\leq 2^{len(\beta)}+B<2^{len(z)+1}$. It is immediate that $len(\alpha)=len(\beta)=len(z)$, so our claim holds.

We can also show that the theory of word equations with regular constraints and $numstr$ predicate is equivalent to the theory $T_{RE,n,c}$. 

For one direction, we need to be able to express an equality predicate between string terms $eq(\alpha,\beta)$, where $\alpha$ and $\beta$ are two string terms. The regular constraints as well as those involving the $numstr$ predicate are canonically encoded. 

This predicate is encoded as follows:\\
$eq(\alpha,\beta) = eqLen(\alpha,\beta) \land numstr(i,1\alpha1\beta) \land numstr(j,1\beta1\alpha) \land (i=j)$. 

Indeed, this tests that $len(\alpha)=len(\beta)$ and $1\alpha1\beta=1\beta1\alpha$. If these are true, it is immediate that $\alpha=\beta$. 

For the converse, it is easy to see that each string constraint $S\in R$ (respectively, $S\notin R$), where $S$ is a string term, can be expressed as the word equation $S=Y_R$, where $Y_R$ is a fresh variable, which is constrained by the regular language defined by $R$ (respectively, by the regular language defined by $\overline{R}$). 

This allows us to define a stronger length-comparison predicate.
We will define a predicate $leqLen(\alpha,\beta)$, where $\alpha$ and $\beta$ are string terms, whose semantics is defined as follows: $leqLen(\alpha,\beta)$ is true if and only if the length of $\alpha$ is smaller than or equal to the length of $\beta$ (i.e., $len(\alpha)\leq len(beta)$). 

We can express $leqLen(\alpha,\beta)$ as:
$leqLen(\alpha,\beta) = (z\in \{0,1\}^*) \land eqLen(\alpha z, \beta)$. 

Finally, we can now move on to $T_{LRE,c,n}$ and show our statement. According to \cite{rp2018strings} the quantifier-free theory of word equations expanded with $numstr$ predicate and length function (not only a length-comparison predicate) and linear arithmetic is undecidable. Thus, if we consider $T_{LRE,n,c}$, this undecidability result immediately holds according to the above.

\qed \end{proof}